\documentclass[journal,10pt,twocolumn]{IEEEtran}
\usepackage{graphicx}
\usepackage{epsfig}
\usepackage{amssymb}
\usepackage{graphics}
\usepackage{algorithm}
\usepackage{algorithmic}
\usepackage{multirow}
\usepackage{amsmath}
\usepackage{subfigure}
\usepackage{setspace}
\usepackage{theorem}
\usepackage{cite}

\newtheorem{theorem}{Property}[section]

\ifCLASSINFOpdf

\else

\fi

\begin{document}
\title{A Geographic Multi-Copy Routing Scheme for DTNs With Heterogeneous Mobility}

\author{Yue Cao,
        Kaimin Wei,
        Geyong Min,~\IEEEmembership{Member,~IEEE},
        Jian Weng,
        Xin Yang and
        Zhili Sun,~\IEEEmembership{Member,~IEEE}

\thanks{Y.Cao, X.Yang and Z.Sun are with the Institute for Communication Systems (ICS), University of Surrey, Guildford, UK. Email: y.cao; xin.yang; z.sun@surrey.ac.uk.}
\thanks{K.Wei (corresponding author) and J.Weng are with the Department of Computer Science, Jinan University, China. Email: cswei@jnu.edu.cn.}
\thanks{G.Min is with the High Performance Computing and Networking (HPCN) research group, University of Exeter, UK. Email: g.min@exeter.ac.uk.}
\thanks{This work is supported by National Natural Science of Foundation of China (Grant No. 61502202), PAPD and CICAEET. Manuscript received on 19 September 2015, revised on 5 January 2016 and 21 March 2016, accepted on 1 May 2016.}
}

\markboth{IEEE Systems Journal}%
{}
\maketitle

\begin{abstract}
Previous geographic routing schemes in Delay/Disruption Tolerant Networks (DTNs) only consider the homogeneous scenario where nodal mobility is identical.
Motivated by this gap, we turn to design a DTN based geographic routing scheme in heterogeneous scenario.
Systematically, our target is achieved via two steps:
1) We first propose ``The-Best-Geographic-Relay (TBGR)'' routing scheme to relay messages via a limited number of copies, under the homogeneous scenario.
We further overcome the local maximum problem of TBGR given a sparse network density, different from those efforts in dense networks like clustered Wireless Sensor Networks (WSNs).
2) We next extend TBGR for heterogeneous scenario, and propose ``The-Best-Heterogeneity-Geographic-Relay (TBHGR)'' routing scheme considering individual nodal visiting preference (referred to non-identical nodal mobility).
Extensive results under a realistic heterogeneous scenario show the advantage of TBHGR over literature works in terms of reliable message delivery, while with low routing overhead.
\end{abstract}


\IEEEpeerreviewmaketitle

\section{Introduction}
\IEEEPARstart{D}{E}lay/Disruption Tolerant Networks (DTNs) \cite{my} can be applied to various scenarios in Wireless Sensor Networks (WSNs) with large size data for transmission, where mobile nodes collect sensing data from others and then offload the data to a sink destination.
Due to either technical or economic reason, it is impossible to establish direct communication path from a data source to a sink destination under such scenarios.
Furthermore, the network may even be partially disconnected, owing to limited transmission range or sparse network density.
In addition to those traditional applications of WSNs \cite{Xie2014,shen}, one of the recent social well-being applications is the implementation of WSNs in Vehicle Ad hoc NETworks (VANETs), also known as Vehicular Sensor Networks (VSNs) \cite{Wen2015395}.

Focusing on the VSNs scenario for delay tolerant based sensing data (or referred to messages) collection, vehicles can help to monitor network (e.g., weather conditions, pollution measurements or road surface conditions) and deliver data to certain deployed destinations.
Different from traditional WSNs where nodes are energy constrained and densely deployed (or with low mobility), vehicles in VSNs are not limited by energy but highly mobile.
In this context, the network topology is highly dynamic and even frequently disconnected.

In DTNs, majority of research efforts have been paid on routing under such intermittently connected scenario.
Different from numerous previous works relying on network topology information to relay messages, geographic routing in DTNs has not received much attention in the literature \cite{my}.
Since it is difficult to obtain the most recent network topology information, instead, relying on geographic information (by tracking where the destination is \cite{zb1,zb2}) to relay messages is feasible.
This has been a valid assumption in traditional Mobile Ad hoc NETworks (MANETs) and WSNs.
It is highlighted that this assumption is also well applicable to VSNs scenario, as the collection points (destinations) are generally located at places where there is high vehicles penetration.
In particular, a recent survey \cite{gsurvey} has reviewed this explicit branch, where related works have already shown the improved routing performance of geographic routing over topological routing in DTNs.

In spite of this motivation, the sparse network density\footnote[1]{In sparse VANETs or VSNs, network nodes are normally connected in a sparse mode and also opportunistically encounter each other. As such, the network would experience frequent network disconnection in rural locations with sparsely populated roads, and also during non-busy hours such as late night.} inevitably has influence on making reliable routing decision and handling the local maximum problem \cite{Karp:2000:GGP:345910.345953}, specifically:
\begin{itemize}
\item In clustered WSNs, a message is geographically relayed to its destination via continuously connected path, thanks to high network density.
    However, this is infeasible in DTNs where a node closer to the destination may move away and even not encounter others in future, due to the sparse network density.
\item  Besides, handling the local maximum problem\footnote[2]{Considered as routing void, this problem implies that if a better relay node is unavailable, the message carrier will keep on carrying the message. In light of this, the message delivery is delayed or even degraded if a better relay node is never met.} in DTNs also faces some new challenges, because of lacking enough encountered opportunities to find relay nodes in sparse networks.
\end{itemize}

Most of routing schemes in DTNs rely on redundancy to improve message delivery, meaning a message is replicated with multiple copies in a network.
The motivation behind this is to enable at least one of these message copies can be delivered successfully before expiration deadline.
As already investigated in \cite{7055214}, limiting the number of message copies (relayed by
the source and possibly other nodes receiving a copy) to $L$\footnote[3]{Note that the initialization of $L$ is application dependent, e.g., based on number of nodes in network to support delay tolerant requirement.}, benefits more from the scenario where nodes are highly mobile.

Concerning mobility, its heterogeneity refers to the behavior that mobile nodes in the same group would have a common preference, to move within a certain area consisting of some popular places.
Whereas this behavior is differentiated among those in diverse groups.
For example, vehicles in VSNs may have daily preferences to deliver data to different collection points.
Inherently, topological routing schemes in DTNs can be applied to both homogeneous and heterogeneous scenarios, because the nodal encounter prediction is based on historical information, e.g., encounter frequency, inter-contact time and encounter duration.
Here, the heterogeneity of mobile nodes can be reflected, because these metrics rely on the fact that ``whether'' pairwise nodes have encountered in the past.

In spite that a few works \cite{6837524,6587056} have attempted geographic utility metric for DTN routing, they generally assume homogeneous scenario where mobile nodes ``will'' encounter destination, however without concerning mobility heterogeneity.
In the worst case, the message would not be relayed to heterogeneous mobile nodes which move within an area (where the destination is located).
By limiting the number of message copies, relaying them to few nodes encountered, implies that some or all of these copies will end up with nodes that may never meet the destination.
With this concern, we address geographic routing in DTNs comprising the mobility heterogeneity, with the following contributions:
\begin{enumerate}
\item Starting from the homogeneous scenario, we first propose The-Best-Geographic-Relay (TBGR) with a limited $L$ number of message copies for delivery.
    Here, based on a given geometric metric for selecting relay node, the routing reliability due to inconsistent nodal motion statuses (in terms of moving direction and activity\footnote[4]{The activity refers to the case that mobile nodes will be temporarily stationary in the network.}) is improved, while the local maximum problem is solved considering the sparse network density.
\item By generalizing the nature of TBGR under the homogeneous scenario, we next turn to the heterogeneous scenario.
    Here, we propose The-Best-Heterogeneity-based-Geographic-Relay (TBHGR) which geographically relays $L$ message copies to the destination, located in the area around which other nodes are heterogeneous.
    Here, we concern the nodal mobility heterogeneity (in terms of visiting preference), in addition to moving direction and activity as investigated in TBGR.
    Along with this, we also manage messages for transmission and storage due to the network resources contention.
\end{enumerate}

The rest of this article is organized as follows.
The related work is presented in Section \uppercase\expandafter{\romannumeral2}.
Then, we present the design as well as the analysis of TBGR in Section \uppercase\expandafter{\romannumeral3}, followed by the design of TBHGR in Section \uppercase\expandafter{\romannumeral4}.
Finally, we conclude our work in Section \uppercase\expandafter{\romannumeral5}.

\section{Related Works}
In the literature, the simplest scheme, Direct Delivery (DD) \cite{1026005}, lets the source node carry messages and deliver them to their destinations.
Although this scheme performs only one transmission, it is extremely slow.
Therefore, other works relay (without replicating any copy through transmission) messages either based on topological \cite{6512845} or geographic utility metrics\cite{1543743,5779182}.
Even if they can achieve a faster delivery than DD, their performance is dramatically degraded in case of sparse network density.
Instead, using redundant message copies has been widely investigated, with the following two main branches, depending on whether or not to limit the number of copies a message can be replicated.

\subsection{Relaying Messages Without Limiting Copy Numbers}
Since Epidemic \cite{Vahdat00epidemicrouting} naively replicates message copies, it only performs well when no contention exists for shared network resources like bandwidth and buffer space.
Many previous works utilize topological utility metrics \cite{prophet} to qualify nodes for selected replication, compared to a few works utilize geographic utility metrics \cite{6837524,6587056}.
To enhance routing efficiency, Delegation Forwarding (DF) \cite{Erramilli:2008:DF:1374618.1374653} enables a message to cache an updated threshold value which is equal to the topological utility metric (in relation to the message destination), and relays message copy to a node (with a better utility metric than this cached threshold).
If without using DF, a node does not keep a threshold value and certainly the message carrier does not update this value after it encounters a better quality node.
While if using DF, a node will raise this threshold value to the quality of a better candidate node, and only uses this threshold value for further comparison.
Thus, with the increase of its level, the replication chance of message carrier is expected to be decreased, which means the number of copies duplicated for a message is also to be reduced.

\subsection{Relaying Messages With Limiting Copy Numbers}
Previous works in this branch assume that when enough nodes in the network are sufficiently mobile, replicating a message  with a limited number of copies is able to achieve an efficient message delivery.
Authors in \cite{4430784} propose Spray-and-Wait (SaW) algorithm, in which a copy ticket $C_M$ is defined for each message, to control how many copies a message can still be replicated.
Depending on an initial value $L$ for $C_M$, Source-SaW (S-SaW) lets source node replicates additional $(L-1)$ copies with a single copy ticket $(C_M=1)$ distributed.
In contrary, Binary-SaW (B-SaW) adopts a binary tree mechanism to equally distribute copy tickets for replicated copy, that is $\frac{C_M}{2}$ rather than $(C_M=1)$.
Here, any intermediate nodes carrying a copy with $(C_M>1)$ copy tickets can also perform message replication.
Note that although to proportionally distribute $C_M$ has been initially proposed in \cite{5061994}, while \cite{4176957} proves that an optimal distribution depends on the value of $L$.
In this context, to equally distribute $C_M$ as adopted in B-SaW, has been widely adopted in literature.
This is because that the way to distribute $C_M$ is independent of utility metric, while an optimal performance is achieved in case of homogeneous nodal mobility only.
Considering the heterogeneous nodal mobility, replicating message copies \cite{4731259} to a better qualified $(L-1)$ nodes has been investigated.
To expedite delivery via topological utility metric, EBSR \cite{7055214} further relays (without generating additional copies) a message copy with $(C_M=1)$ copy ticket.
Based on geographic utility metric and underlying map topology, GeoSpray \cite{geospray} (via underlying map topology) calculates the Nearest Point (NP) to destination to guide message relay.

\subsection{Motivation}
All reviewed geographic routing schemes in DTNs assume homogeneous scenario.
However, existing homogenous model defines identical mobility pattern that cannot completely and realistically depict the heterogeneous scenario.
In the latter case, the nodal mobility is limited within a certain area, as such relaying a message across different areas may suffer from a high overhead (relaying the message to nodes which never encounter destination) or even failure (the message will not be relayed to nodes moving towards the area where the destination is located).
In this article, we tackle geographic routing scheme in DTNs comprising heterogeneous nodal mobility.

\section{Design of TBGR}
\subsection{Assumption}
We assume that the Global Position System (GPS) is available for all mobile nodes in the network, while the energy consumption for running GPS could be neglected in VSNs.
The location of the stationary sink destination is available for all mobile nodes under a two dimensional scenarios, we also allow multiple sink destinations coexisting in a network.
Although the influence of location error inevitably affects the accuracy of routing decision, our research effort is for improving the message delivery based on any delivery time based metric.
In other words, the nature of TBGR as discussed herein is still applicable to the condition where location error exists.
Here, the encounter possibility between pairwise nodes is identical under a homogeneous network.

A slotted based collision avoidance MAC protocol is applied by each node for contention resolution, such that only one connection can be established between two neighbor nodes at each time slot.
Therefore, any node can not perform transmission to more than one node within each slot.
Also, in networks that are quite sparse, we expect that only a few nodes would be close enough each time to compete for bandwidth.
Certain media access algorithm \cite{Mostafa:2014:RAP:2580129.2580647} will determine and limit the accessing media and its duration (allocation time) within an encounter duration.

Although we envision for delay tolerant based data collection in VSNs, messages are usually with a certain lifetime, evaluated as Time-To-Live (TTL).
Within a message lifetime, a limited $L$ number of copies including $(L-1)$ replicated copies and the original message will exist in a network, such that at least one of them can be successfully delivered by the sink destination.
Otherwise, the message will be discarded until TTL expires.

\begin{table}[htbp]\scriptsize
\renewcommand{\arraystretch}{1.3}
\caption{List of Notations Defined in TBGR}\vspace{-10pt}
\label{configure}
\centering
\begin{tabular}{|p{0.7cm}|p{6.6cm}|}
\hline
$N_i$ & Message carrier\\\hline
$N_j$ & Encountered node\\\hline
$N_d$ & Destination of message\\\hline
$M$ & Message carried by $N_i$\\\hline
$T^{ini}_M$ & Initial message lifetime\\\hline
$T^{ela}_M$ & Elapsed time since message generation\\\hline
$D_{i,d}$ & Distance from $N_i$ to $N_d$, and similarly for $D_{j,d}$\\\hline
$\phi_{i,d}$ & Relative angle between the moving direction of $N_i$ and $D_{i,d}$, and similarly for $\phi_{j,d}$\\\hline
$S_i$ & Moving speed of $N_i$, similarly for $S_j$\\\hline
$R$ & Transmission range of device\\\hline
$V^T_M$ & Heuristic delivery time based threshold value cached in $M$\\\hline
$C_M$ & Copy ticket of $M$\\\hline
$L$ & Initially defined value for $C_M$ \\\hline
\end{tabular}
\label{table1}
\end{table}

\subsection{Geographically Relaying $L$ Message Copies}
As described in TABLE \ref{table1}, we denote $N_i$ as the message carrier, $N_j$ is its encountered node, while $N_d$ is denoted as the destination for message $M$.
Based on an encounter between $N_i$ and $N_j$, a routing decision is made based on whether $N_j$ is a ``better'' option to carry a message copy.
Note that the original message is allowed for replicating with additional $(L-1)$ copies, based on a pre-defined value $L$.

Here, the copy ticket $C_M$ is defined as an additional flag in each message, to control how many copies a message can still be replicated.
Initially, $C_M$ is set to $L$ for any newly generated message.
Since the effectiveness to equally distribute $C_M$ has been analyzed in \cite{4430784} and widely adopted by previous works, we thereby apply this mechanism in this article.
Further concerning on how to proportionally distribute message copy tickets based on a utility metric is out of our discussion.

To generalize the nature of TBGR, we herein select an existing utility metric \cite{5779182} which jointly considers nodal distance, moving speed and direction to capture nodal homogeneous mobility.
In Fig.\ref{ta}, this utility metric is calculated as $\frac{D_{i,d}-R}{S_i\times\cos\phi_{i,d}}$ where $\phi_{i,d}$ (refer to its calculation in \cite{5982776}) is the relative angle between the moving direction of $N_i$ and the distance $D_{i,d}$ measured from $N_i$ to $N_d$.
Since a smaller $\phi_{i,d}$ implies a smaller diversity to $N_d$, this metric intends to capture the time for $N_i$ to intersect $N_d$ given a fast moving speed $S_i$.

Upon the efficient message delivery via a limited number of message copies, we propose that $N_i$ will relay a message copy to $N_j$, and equally distribute the value of $C_M$ for the message carried by $N_i$ and the copy replicated to $N_j$, given that:
\begin{equation}\footnotesize
\left(\frac{D_{i,d}-R}{S_i\times\cos\phi_{i,d}}>\frac{D_{j,d}-R}{S_j\times\cos\phi_{j,d}}\right)\cap\left(\phi_{i,d}<\frac{\pi}{2}\right)
\cap\left(\phi_{j,d}<\frac{\pi}{2}\right)
\label{con1}
\end{equation}
This implies that $N_j$ has a better delivery potential than $N_i$, and is given with a message copy with $\frac{C_M}{2}$ copy tickets for further dissemination at upcoming encounter opportunities.

\begin{figure}[htbp]
\begin{center}
\includegraphics[scale=0.38]{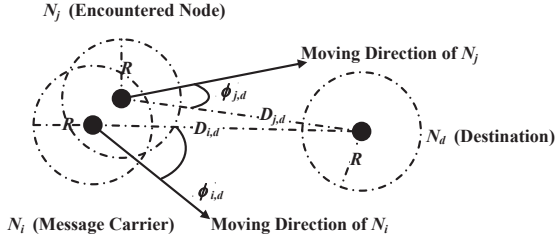}
\caption{Illustration of Geometric Metric}\vspace{-10pt}
\label{ta}
\end{center}
\end{figure}

However, we note that this utility metric has limitation.
This is because to calculate $\frac{D_{i,d}-R}{S_i\times\cos\phi_{i,d}}$ is invalid given $\left(\phi_{i,d}\geq\frac{\pi}{2}\right)$, for which $\frac{D_{i,d}-R}{S_i\times\cos\phi_{i,d}}$ must be a positive value.
In light of this, the above routing decision has limitation when the message carrier is moving away from its destination, even if its encountered node is moving towards its destination.
Upon the nature of DF adopting a threshold value to qualify the encountered node, this operation motivates us to overcome the above limitation, since the quality of message carrier is not required for making routing decision.
In light of this, we utilize DF by caching the threshold value $V^T_M$ for each message, and then convert the conditions \eqref{con1} into:
\begin{equation}\footnotesize
\left(V^T_M>\frac{D_{j,d}-R}{S_j\times\cos\phi_{j,d}}\right)\cap\left(\phi_{j,d}<\frac{\pi}{2}\right)
\label{con2}
\end{equation}
Note that $V^T_M$ is an updated value for recording $\left(\frac{D_{j,d}-R}{S_j\times\cos\phi_{j,d}}\right)$, as an additional flag defined in the message.
Consequently, as illustrated between lines 8 and 11 in Algorithm \ref{dgs1}, if given $\left(V^T_M>\frac{D_{j,d}-R}{S_j\times\cos\phi_{j,d}}\right)\cap\left(\phi_{j,d}<\frac{\pi}{2}\right)$, $N_i$ will set $\frac{C_M}{2}$ copy tickets for the replicated copy to $N_j$, and keeps the rest $\left(C_M-\frac{C_M}{2}\right)$ copy tickets for the message carried by itself.

Here, $V^T_M$ is updated to $\frac{D_{j,d}-R}{S_j\times\cos\phi_{j,d}}$ given $\left(V^T_M>\frac{D_{j,d}-R}{S_j\times\cos\phi_{j,d}}\right)\cap\left(\phi_{j,d}<\frac{\pi}{2}\right)$.
This value is recorded as the heuristic delivery time of the historically encountered node, to compare with that of upcoming encountered node.
Therefore, the limitation in relation to the moving direction of $N_i$ is overcome, by comparing the utility metric of historically encountered node with that of upcoming encountered node.

Although the above decision may also select the node of which $\phi_{j,d}$ is close to $\frac{\pi}{2}$, the routing decision will be towards optimality.
This is because that apart from the gradually decreased $D_{j,d}$, to gradually update $V^{T}_M$ implies either the smallest $\phi_{j,d}$ or largest $S_j$.
In contrast, when using conditions \eqref{con1}, $N_j$ might be selected to carry $M$, given that $\phi_{j,d}$ is close to $\frac{\pi}{2}$ while $S_j$ is large.
Consequently, although $M$ has been relayed for several copies, it may not be close to $N_d$.
In this context, in TBGR the quality of the relay node will be gradually converged and the ``best\footnote[5]{In addition to concern inconsistent nodal moving direction, TBGR can be extended to other condition such as nodal activity. This is further discussed in analysis section.}'' relay node will be found.

We also make modification by initializing $V^T_M$ with an infinitely large value for each message, rather than $\frac{D_{i,d}-R}{S_i\times\cos\phi_{i,d}}$.
The motivation is to overcome the limitation if $M$ is generated in $N_i$, but it is moving away from $N_d$ as given by $\left(\phi_{i,d}\geq\frac{\pi}{2}\right)$.
Furthermore, if both $N_i$ and $N_j$ have a copy of the same message, it is essential to update $V^T_M$ towards a smaller value for these two copies.
This operation is performed in a distributed manner, and only happens when there is an encounter between two nodes holding copies of the same message.

\subsection{Local Maximum Problem Handling}
The local maximum problem happens if $N_j$ does not meet conditions $\left(V^T_M>\frac{D_{j,d}-R}{S_j\times\cos\phi_{j,d}}\right)\cap\left(\phi_{j,d}<\frac{\pi}{2}\right)$.
In this case, we propose to still keep on replicating a copy with equally distributed $C_M$, if $\frac{D_{i,d}-R}{S_i\times\cos\phi_{i,d}}$ is longer than the TTL, given:
\begin{equation}\footnotesize
\left(\frac{D_{i,d}-R}{S_i\times\cos\phi_{i,d}}>T^{ini}_M-T^{ela}_M\right)
\label{cc}
\end{equation}
Here, $T^{ela}_M$ is denoted as the elapsed time since message generation, while $T^{ini}_M$ is denoted as the initialized message lifetime.
Thus, $(T^{ini}_M-T^{ela}_M)$ is calculated as the TTL of message.

Since for message delivery before its TTL, $\frac{D_{i,d}-R}{S_i\times\cos\phi_{i,d}}$ should not be longer than $\left(T^{ini}_M-T^{ela}_M\right)$.
Therefore, giving a message copy to $N_j$ with $\frac{C_M}{2}$ copy tickets distributed, is thus executed when $\left(\frac{D_{i,d}-R}{S_i\times\cos\phi_{i,d}}>T^{ini}_M-T^{ela}_M\right)$.
This hopes that $N_j$ would encounter other nodes that with a shorter time to intersect $N_d$ at the upcoming encounter.

Furthermore, we convert conditions $\left(\frac{D_{i,d}-R}{S_i\times\cos\phi_{i,d}}>T^{ini}_M-T^{ela}_M\right)\cap\left(\phi_{i,d}<\frac{\pi}{2}\right)$ into:
\begin{equation}\footnotesize
\left(V^T_M>T^{ini}_M-T^{ela}_M\right)\
\label{local1}
\end{equation}
It is observed that, by removing the condition $\left(\phi_{i,d}<\frac{\pi}{2}\right)$ for calculating $\frac{D_{i,d}-R}{S_i\times\cos\phi_{i,d}}$, the condition \eqref{local1} applied for decision between lines 12 and 14 in Algorithm \ref{dgs1}, implies that $V^T_M$ as the smallest value of delivery time recorded in the network, is still longer than the TTL.
Note that updating $V^T_M$ is excluded in this case, due to $\left(V^T_M\leq\frac{D_{j,d}-R}{S_j\times\cos\phi_{j,d}}\right)$.
Meanwhile, since the value of $V^T_M$ is not with an infinitely large value due to $\left(V^T_M\leq\frac{D_{j,d}-R}{S_j\times\cos\phi_{j,d}}\right)$, using the historically lowest value of $\frac{D_{j,d}-R}{S_j\times\cos\phi_{j,d}}$ to compare with $(T^{ini}_M-T^{ela}_M)$, adequately enhances the reliability for message delivery depending on TTL.

\subsection{Computation Complexity of Routing Decision}
The computation complexity of TBGR is $O(\log_2L)$ to have $L$ message copies with $(C_M=1)$ copy ticket.
This is because copy tickets is distributed by a binary manner, and TBGR chooses $(L-1)$ relay nodes based
on the $V^T_M$ in descending order.

\begin{algorithm}[htbp]\footnotesize
\caption{Routing Process of TBGR}
\label{dgs1}
\begin{algorithmic}[1]
\STATE set $(C_M=L)$
\STATE set $V^T_M$ with an infinitely large value
\FOR{each encounter between $N_i$ and $N_j$}
\FOR{each $M$ carried by $N_i$}
\IF {$N_j$ already has a copy of $M$}
\STATE update $V^T_M$ for $M$ including its copy carried by both $N_i$ and $N_j$
\ELSIF {$(C_M>1)\cap(S_j\neq0)$}
\IF {$\left(V^T_M>\frac{D_{j,d}-R}{S_j\times\cos\phi_{j,d}}\right)\cap\left(\phi_{j,d}<\frac{\pi}{2}\right)$}
\STATE update $V^T_M$ towards $\frac{D_{j,d}-R}{S_j\times\cos\phi_{j,d}}$
\STATE replicate $M$ to $N_j$ with $\frac{C_M}{2}$ copy tickets
\STATE keep $\left(C_M-\frac{C_M}{2}\right)$ copy tickets for $M$ in $N_i$
\ELSIF {$\left(V^T_M>T^{ini}_M-T^{ela}_M\right)$}
\STATE replicate $M$ to $N_j$ with $\frac{C_M}{2}$ copy tickets
\STATE keep $\left(C_M-\frac{C_M}{2}\right)$ copy tickets for $M$ in $N_i$
\ENDIF
\ENDIF
\ENDFOR
\ENDFOR
\end{algorithmic}
\end{algorithm}

\subsection{Analysis}
We consider the Random WayPoint (RWP) mobility model, where the meeting times of nodes are assumed to be Independent and Identically Distributed (IID) exponential random variables.
This follows homogeneous nodal mobility as assumed for TBGR.
Here, the number of nodes in a network is denoted as $K$.
The contention of the limited bandwidth and buffer space is not taken into account, as we are mainly concerned with the routing nature of TBGR.

The version that TBGR without local maximum handling but with different policies for distributing copy tickets $C_M$, are named as S-TBGR which distributes $(C_M=1$) and B-TBGR which distributes $\frac{C_M}{2}$ respectively.
Here, we name the S-TBGR without improving routing reliability and local maximum handling, as S-A-Better-Geographic-Relay (S-ABGR).

\begin{theorem}
S-TBGR achieves a higher message delivery ratio than S-ABGR.
\label{t1}
\end{theorem}
\begin{proof}
Considering there is single message which can be replicated with additional $(L-1)$ copies, its delivery probability $P_M$ is calculated
as:
\begin{equation}\footnotesize
P_M=1-\left(1-P_L\right)^L=1-\left(1-P_{(L-1)}\right)^{(L-1)}\times\frac{1}{K-1}
\end{equation}
where $L$ is the number of message copies including the original message, meanwhile $P_{(L-1)}$ is the probability that each replicated copy is successfully relayed from the source to destination.
This equation shows a larger $L$ and $P_{(L-1)}$ lead to a higher $P_M$, based on a
given number of nodes $K$ in a network.
Considering the RWP mobility model, $\frac{1}{K-1}$ is the probability that the original message is directly delivered if the message carrier encounters destination. It also equals to the probability that any pairwise nodes encounter.

\begin{figure}[htbp]
\begin{center}
\includegraphics[scale=0.43]{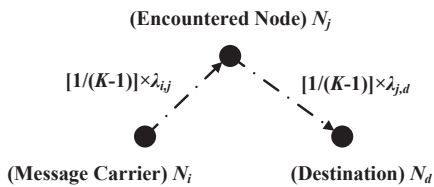}
\caption{Message Delivery Process in S-TBGR}\vspace{-10pt}
\label{2hp}
\end{center}
\end{figure}

Since each message copy is delivered within two hops shown in Fig.\ref{2hp}, message copies are relayed to the first $(L-1)$ encountered nodes $N_j$.
Then these messages copies are only delivered when the destination $N_d$ is in proximity.
In this case, $P_{(L-1)}$ is calculated as:
\begin{equation}\footnotesize
P_{(L-1)}=\lambda_{i,j}\times\lambda_{j,d}\times\left(\frac{1}{K-1}\right)^2
\end{equation}
Here, we have the relaying possibility from $N_j$ to $N_d$, as $(\lambda_{j,d}=1)$, since message is directly delivered when $N_d$ is met.

Recall that the routing decision of S-ABGR follows conditions \eqref{con1},  while that of S-TBGR follows conditions \eqref{con2}.
Here, S-TBGR has a higher value of $\lambda_{i,j}$ than S-ABGR.
This is because that given $\left(\phi_{j,d}<\frac{\pi}{2}\right)$, the probability to meet $\left(V^T_M>\frac{D_{j,d}-R}{S_j\times\cos\phi_{j,d}}\right)$ is 1 because $V^T_M$
is initialized with an infinitely large value, which is larger than the probability that $\left(\frac{D_{i,d}-R}{S_i\times\cos\phi_{i,d}}>\frac{D_{j,d}-R}{S_j\times\cos\phi_{j,d}}\right)
\cap
\left(\phi_{i,d}<\frac{\pi}{2}\right)$ as adopted in S-ABGR.
Therefore, S-TBGR achieves a higher delivery ratio due to a larger $P_{(L-1)}$.
\end{proof}

\begin{theorem}
S-TBGR achieves a lower delivery delay than S-ABGR.
\label{t2}
\end{theorem}
\begin{proof}
The delivery delay of S-ABGR follows:
\begin{equation}\footnotesize
Delay_{(DD)}>Delay_{(S-ABGR)}>Delay_{(S-SaW)}
\end{equation}
In one case, if any encountered node does not meet conditions \eqref{con1}, the message is delivered only when the destination is in proximity.
This is the same as DD where the message is never relayed by any intermediate node, because it is only delivered when the destination is in proximity.
In another case, the message will be relayed if any encountered node meets conditions \eqref{con1}.
This follows S-SaW where a message will be replicated to the first $(L-1)$ encountered nodes.
By means of this, the delivery delay of DD and S-SaW are the lower and upper bounds for that of S-ABGR.
Since S-TBGR achieves a higher delivery probability than S-ABGR based on Property \ref{t1}, the former thereby is with a lower delivery delay due to a faster message copies dissemination.
\end{proof}

\begin{theorem}
B-TBGR achieves a lower delivery delay than S-TBGR.
\label{t3}
\end{theorem}
\begin{proof}
Referring to \cite{4430784}, the delivery delay of S-SaW is given by:
\begin{equation}\footnotesize
Delay_{(S-SaW)}=\sum^{L-1}_{H=1}\frac{EMT}{K-H}+\left(\frac{K-L}{K-1}\right)\times EW
\end{equation}
where $H$ is the distribution depth as shown in Fig.\ref{sprayprocess}, $EMT$ is the expected meeting time under the RWP mobility model.
Meanwhile, $\left(EW=\frac{EMT}{L}\right)$ is defined as expected duration that a message copy is delivered if the destination is in proximity.

\begin{figure}[htbp]
\begin{center}
\includegraphics[width=8cm,height=4.3cm]{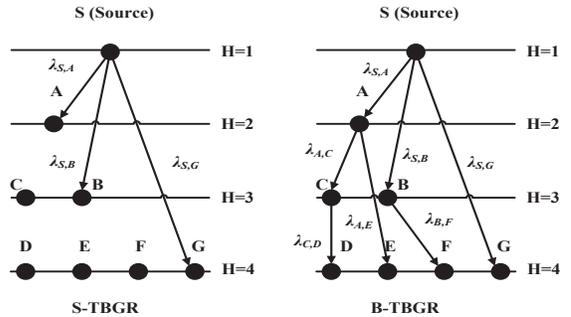}
\caption{Distributing Message Copies in S-TBGR and B-TBGR}\vspace{-10pt}
\label{sprayprocess}
\end{center}
\end{figure}

Note that the expected time until one of another $(L-H)$ message copy is distributed, equals to the time until the source node meets one of $(K-H)$ nodes which does not have this message. Therefore, the total delay for distributing $(L-1)$ copies is $\sum^{L-1}_{H=1}\frac{EMT}{K-H}$.
Meanwhile, $\left(\left(1-\frac{L-1}{K-1}\right)\times EW=\left(\frac{K-L}{K-1}\right)\times EW\right)$ equals to the delay for encountering destination, meaning that destination is not in the first $(L-1)$ encountered nodes.

Rather than replicating a message copy with $(C_M=1)$ in S-SaW, a copy is conditionally replicated in S-TBGR if the encountered node meets conditions \eqref{con2}.
Here, we denote a general relaying possibility $\lambda\in[0,1]$ to simplify analysis, where $\lambda$ might be different by comparing the geographic utility metrics between pairwise nodes.
As such, we have:
\begin{equation}\footnotesize
Delay_{(S-TBGR)}=\sum^{L-1}_{H=1}\frac{EMT}{\lambda\times(K-H)}+\left(\frac{K-L}{K-1}\right)\times EW
\end{equation}
Next, we extend above derivation for $Delay_{(B-TBGR)}$, as:
\begin{equation}\footnotesize
Delay_{(B-TBGR)}=\sum^{H_{max}}_{H=1}\frac{EMT}{\lambda\times2^{H-1}\times \left(K-2^{H-1}\right)}+\left(\frac{K-L}{K-1}\right)\times EW
\end{equation}
This is because the number of nodes which have message copies is $2^{H-1}$ at the $H$ distribution depth.
Then the time for one of these $2^{H-1}$ nodes, relays a message copy to one of $\left(K-2^{H-1}\right)$ nodes (without message) but meets condition \eqref{con2}, is calculated as $\frac{EMT}{\lambda\times2^{H-1}\times\left(K-2^{H-1}\right)}$.
Note that $\left(H_{max}<L-1\right)$ is the maximum depth required for distributing $(L-1)$ message copies.
Given $(H_{max}>2)$ that the message has not been delivered, then we observe $Delay_{(B-TBGR)}$ is lower than $Delay_{(S-TBGR)}$ given Fig.\ref{sprayprocess}.
In other words, B-TBGR distributes more message copies than S-TBGR given the same depth, which contributes to a lower delivery delay.
\end{proof}

\begin{theorem}
The proposed approach for handling the local maximum problem enhances the routing reliability and reduces the delivery delay.
\label{t4}
\end{theorem}
\begin{proof}
By jointly referring to Property \ref{t2} and Property \ref{t3}, a faster distributing $C_M$ until $(C_M=1)$ achieves a higher delivery probability and a lower delivery delay, with an expense of higher overhead.
However, since the total number of message copies is limited, TBGR outperforms B-TBGR if the TTL is finite.
\end{proof}

\begin{theorem}
The routing decision in TBGR is tolerant to nodal activity.
\label{t5}
\end{theorem}
\begin{proof}
we further discuss our investigation regarding DF, where S-ABGR also does not consider the case that the message carrier $N_i$ is temporarily stationary even if the encountered node $N_j$ is good to relay messages.

In this case, the conditions that S-ABGR replicates messages follow:
\begin{equation}\footnotesize
\begin{split}
&\left(\frac{D_{i,d}-R}{S_i\times\cos\phi_{i,d}}>\frac{D_{j,d}-R}{S_j\times\cos\phi_{j,d}}\right)\cap\left( \phi_{j,d}<\frac{\pi}{2}\right)\cap\left(\phi_{i,d}<\frac{\pi}{2}\right)\\
&\cap\left(S_i\neq0\right)\cap\left( S_j\neq 0\right)
\label{c9}
\end{split}
\end{equation}
The conditions that S-TBGR replicates messages follow:
\begin{equation}\footnotesize
\left(V^T_{M}>\frac{D_{j,d}-R}{S_j\times\cos\phi_{j,d}}\right)\cap\left(\phi_{j,d}<\frac{\pi}{2}\right)\cap\left(S_j \neq0\right)
\label{c10}
\end{equation}
In light of this, the limitation regarding the temporarily stationary state of $N_i$ is also overcome.
This is applicable to TBGR since the way to distribute $C_M$ is independent of utility metric.
It is highlighted that any utility metric related to a time to intersect the destination, via either on-hop based prediction as applied in TBGR, or further multi-hops based prediction, is still applicable to TBGR.
As such, TBGR can be applied to the network with various degrees on network density.
\end{proof}

\subsection{Performance Evaluation}

So far, we have discussed a number of properties for TBGR.
We use the Opportunistic Network Environment (ONE) \cite{onesimulator} version 1.4.1 for evaluation, under the RWP scenario with 1000$\times$1000 $m^2$ area, consisting of 1 stationary destination located at the center of scenario and 100 mobile nodes with the constant 5 m/s moving speed.
Note that this is a homogeneous scenario where nodal encounter is identical.
The communication technique is set with 10m transmission range and 2 Mbit/s bandwidth, considering as the Bluetooth standard.

Here, we have $(L=10)$ as the initial value of copy ticket, based on the discussion in \cite{4430784} that choosing $L$ equals to around 10$\%$ number of mobile nodes in a network.
Messages are randomly generated from mobile nodes for every 30s, with 1KB size and default 60 minutes lifetime.
The configuration guarantees that there is no resource contention for bandwidth and buffer space.
To measure the full activity of a network, the message generation ends up before 18000s with additional 3600s to consume unexpired messages.
We run the entire simulation for 10 times and plot the results with 95$\%$ confidence interval.
\begin{figure*}[htbp]
  \centering
  \subfigure[Delivery Ratio]
  {
  \centering
    \label{f11}
    \includegraphics[width=5.7cm,height=3.3cm]{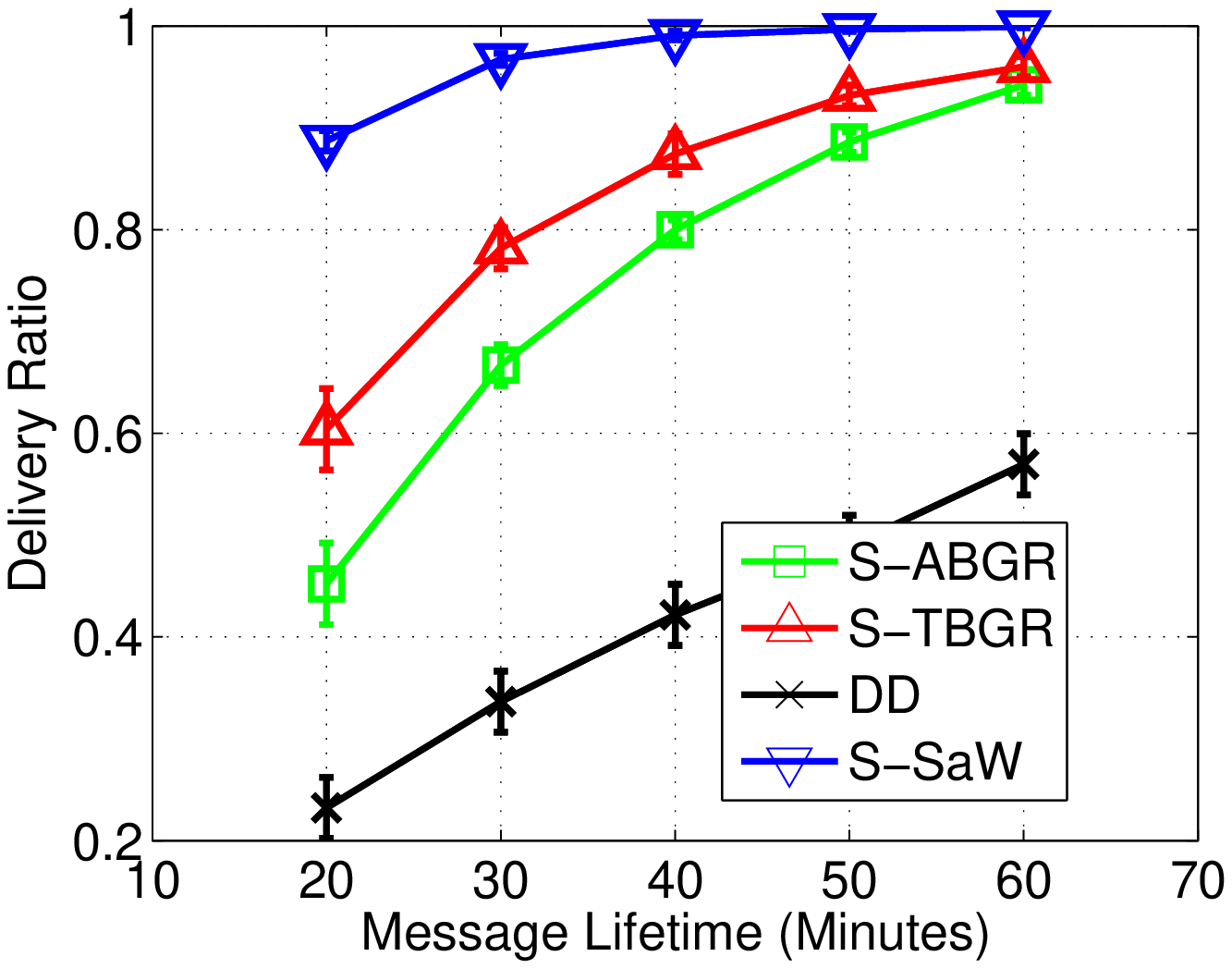}
    }
  \subfigure[Average Delivery Latency]
  {
  \centering
  \label{f12}
  \includegraphics[width=5.7cm,height=3.3cm]{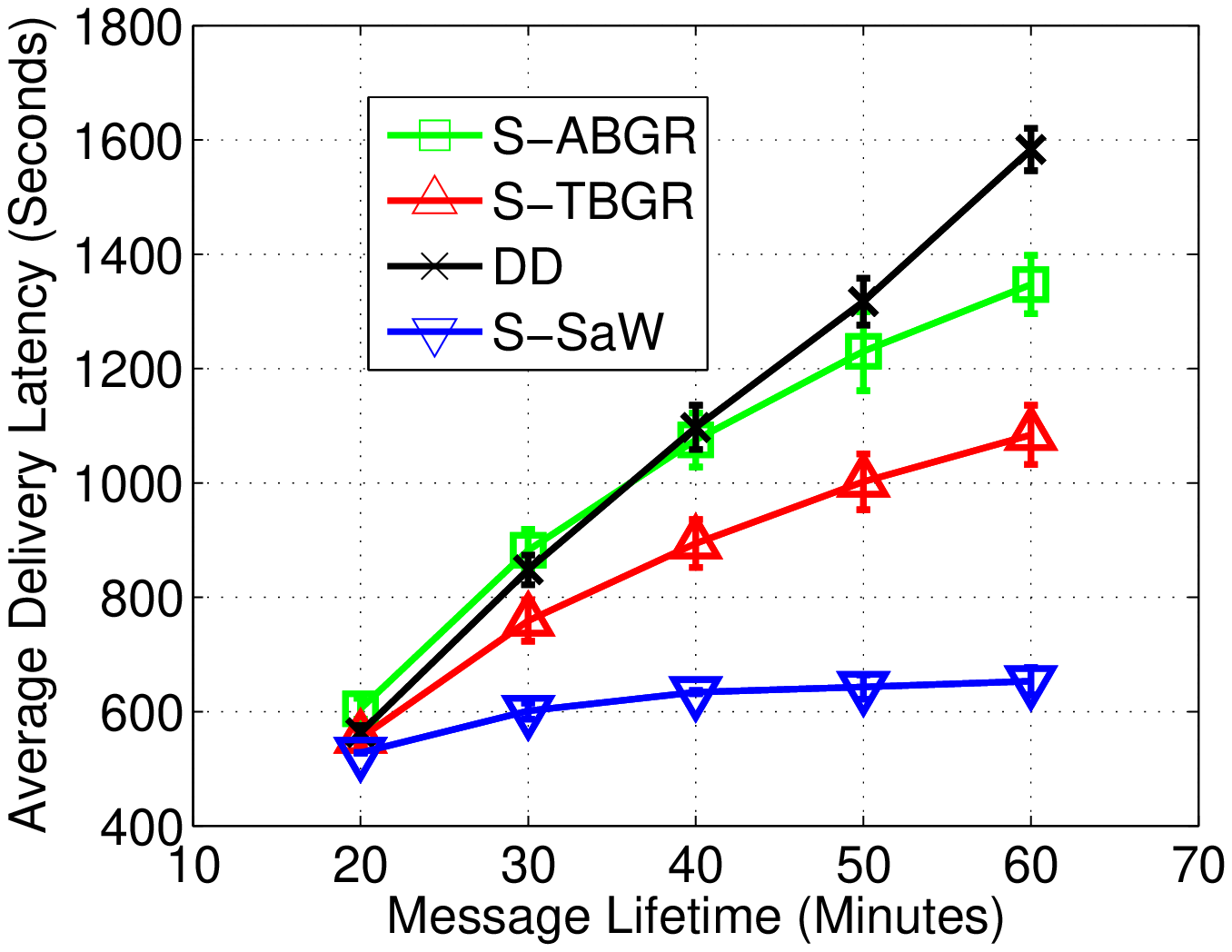}
  }
  \subfigure[Overhead Ratio]
  {
  \centering
  \label{f13}
  \includegraphics[width=5.7cm,height=3.3cm]{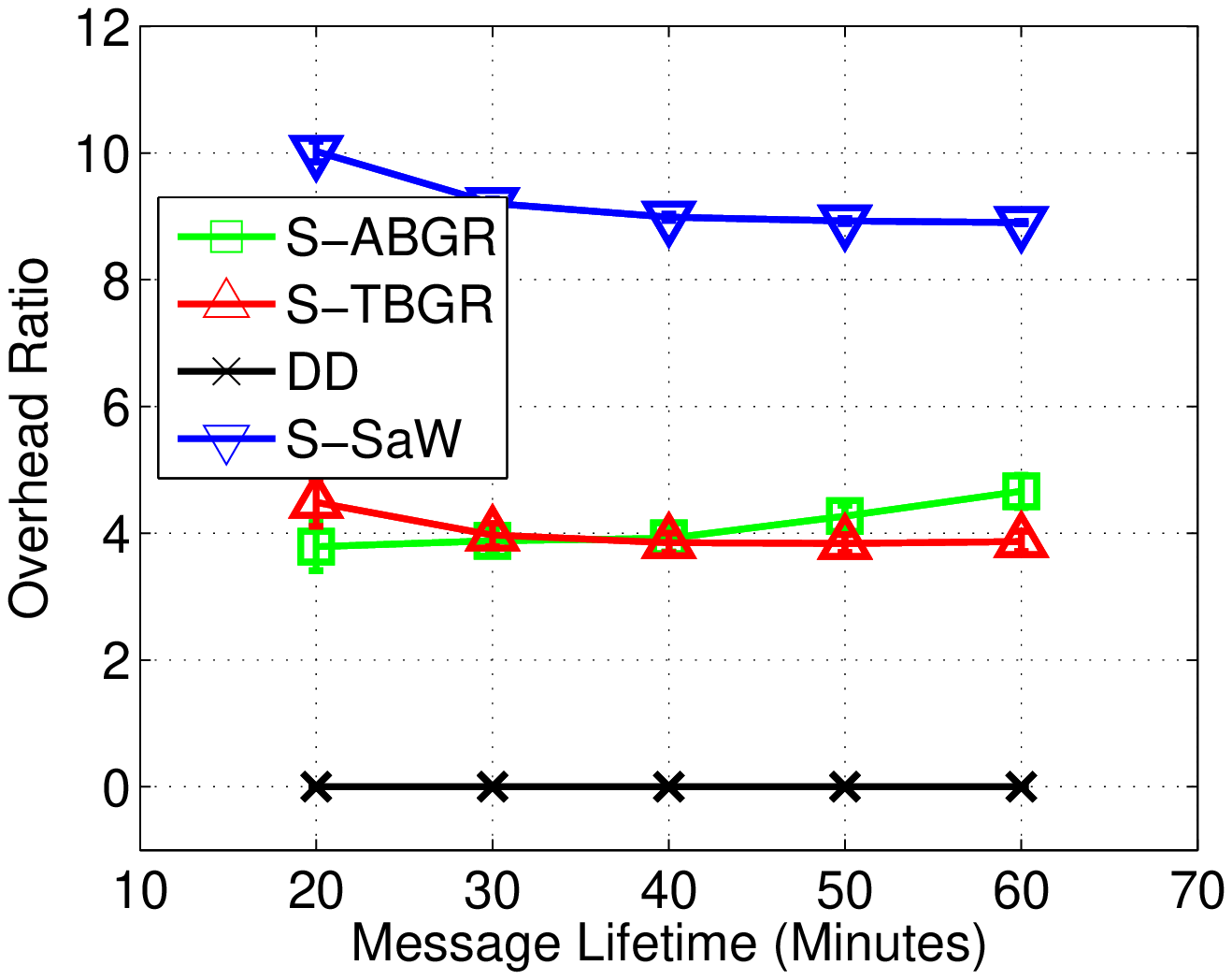}
  }
  \caption{Influence of Message Lifetime Given 10m Transmission Range and 5 m/s Moving Speed}\vspace{-10pt}
\end{figure*}

\begin{figure*}[htbp]
  \centering
  \subfigure[Moving Speed]
  {
  \centering
    \label{f21}
    \includegraphics[width=5.7cm,height=3.3cm]{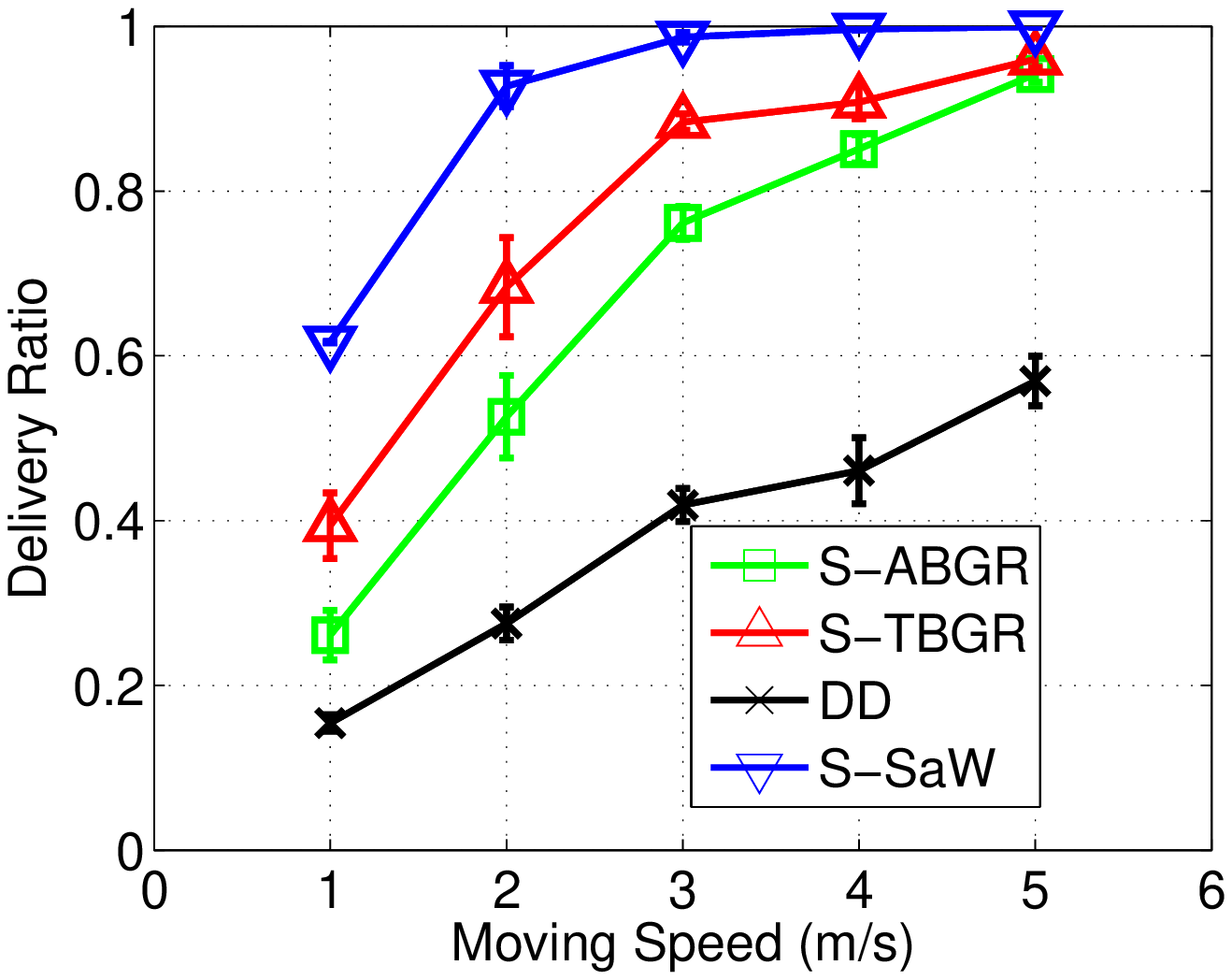}
    }
  \subfigure[Transmission Range]
  {
  \centering
    \label{f22}
    \includegraphics[width=5.7cm,height=3.3cm]{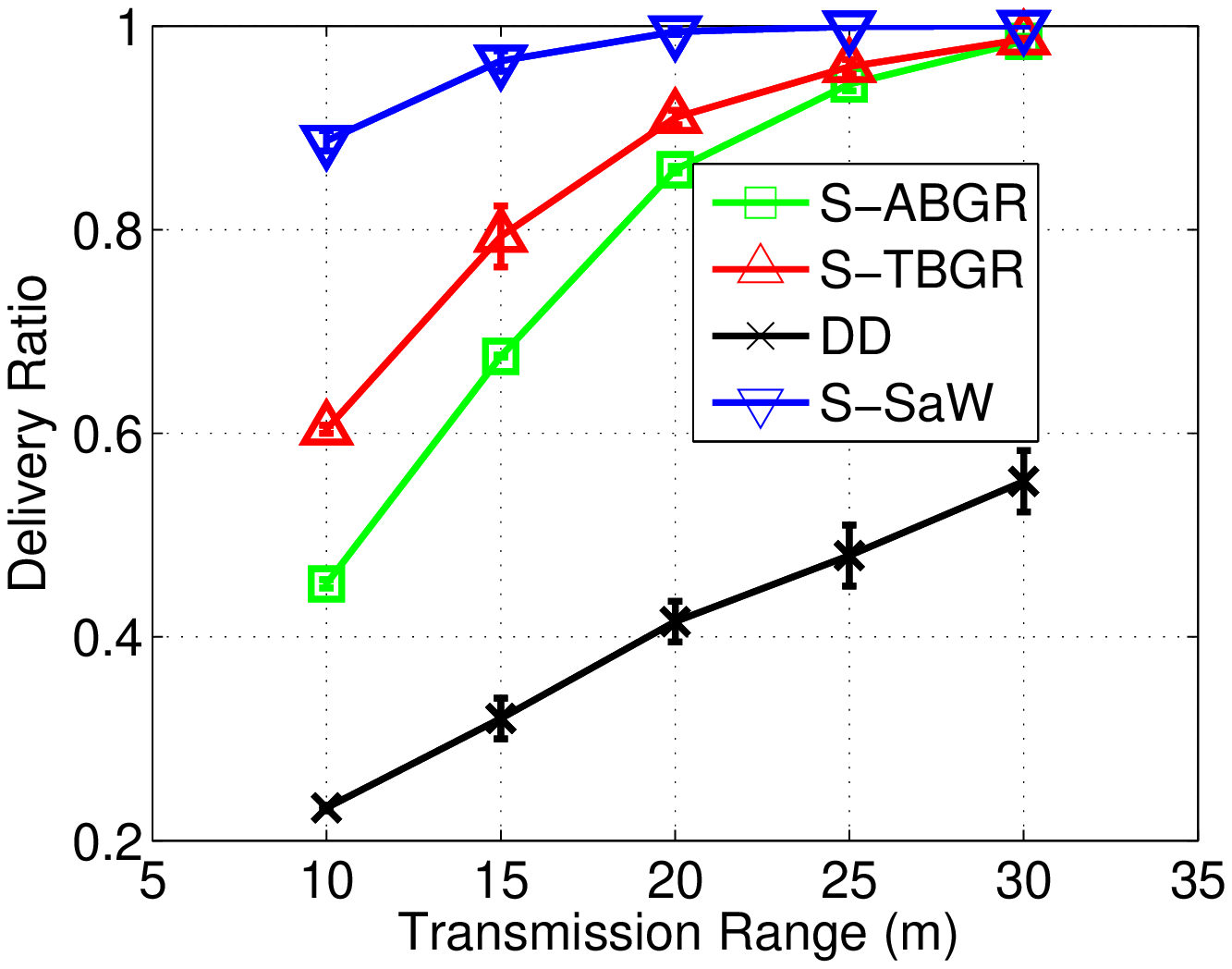}
    }
    \subfigure[Maximum Waiting Time]
    {
  \centering
    \label{f23}
    \includegraphics[width=5.7cm,height=3.3cm]{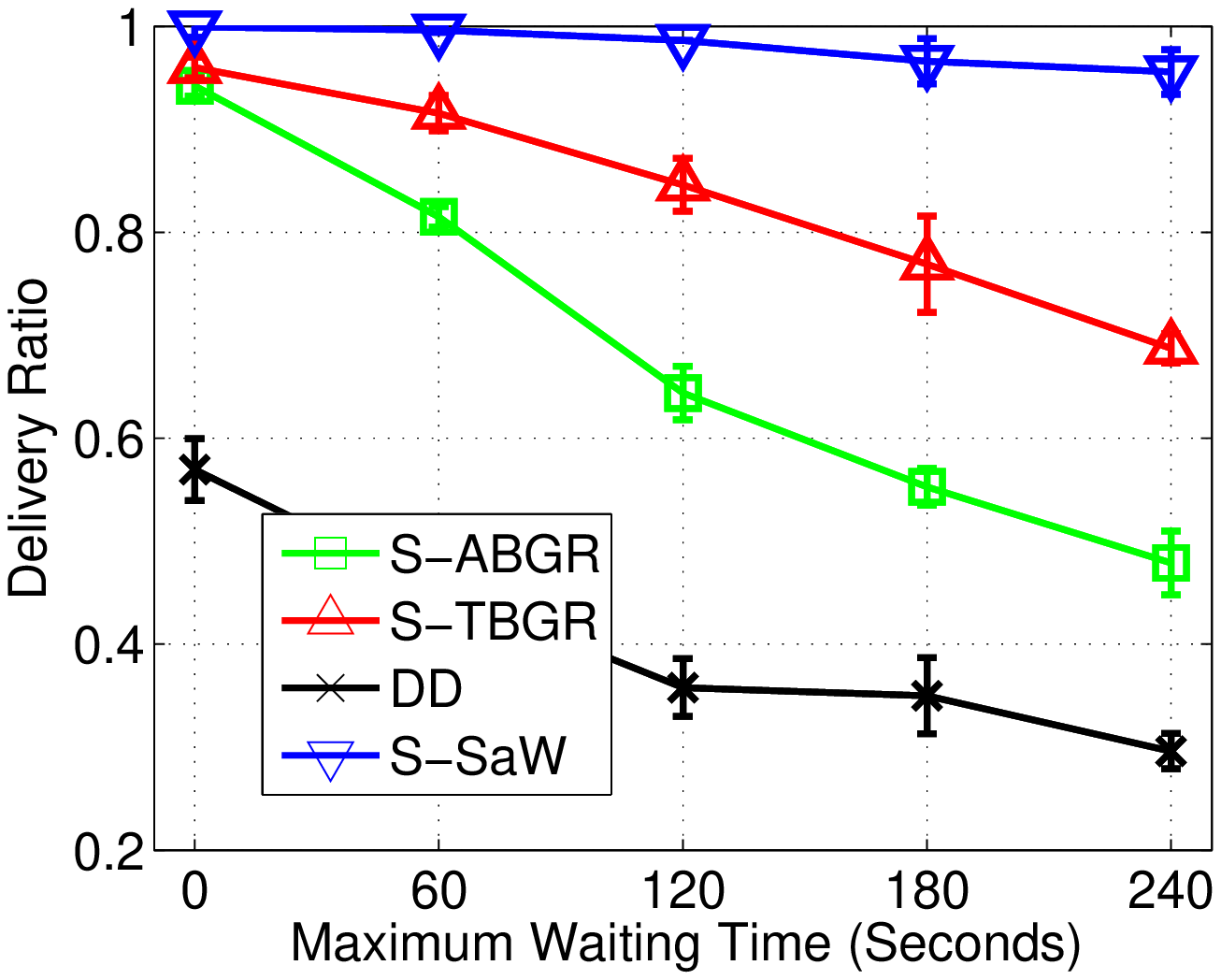}
    }
  \caption{Influence of Moving Speed, Transmission Range and Maximum Waiting Time}\vspace{-10pt}
  \label{2fig}
\end{figure*}

\begin{figure*}[htbp]
  \centering
  \subfigure[Delivery Ratio]
  {
  \centering
    \label{f31}
    \includegraphics[width=5.7cm,height=3.3cm]{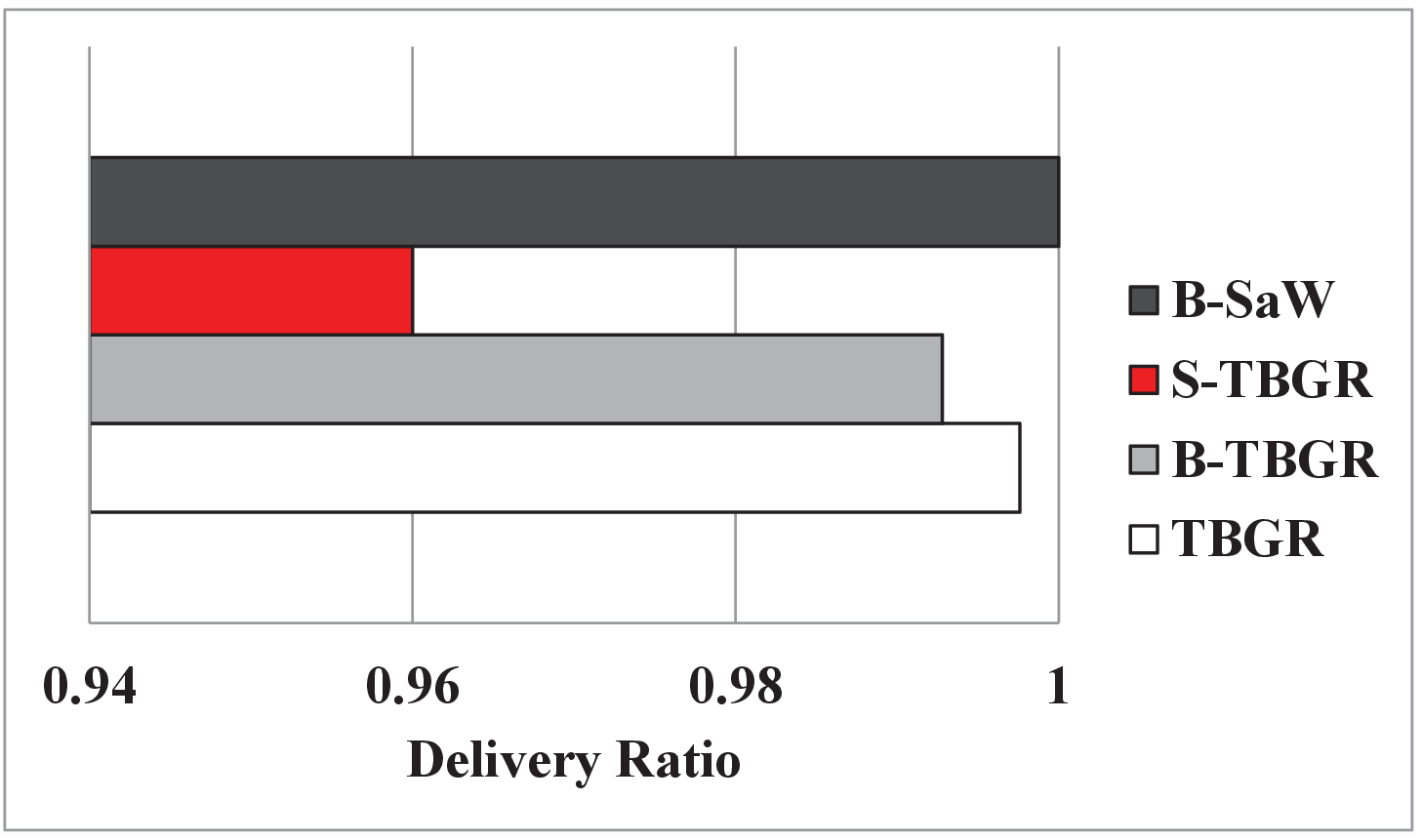}
    }
  \subfigure[Average Delivery Latency]
  {
  \centering
    \label{f32}
    \includegraphics[width=5.7cm,height=3.3cm]{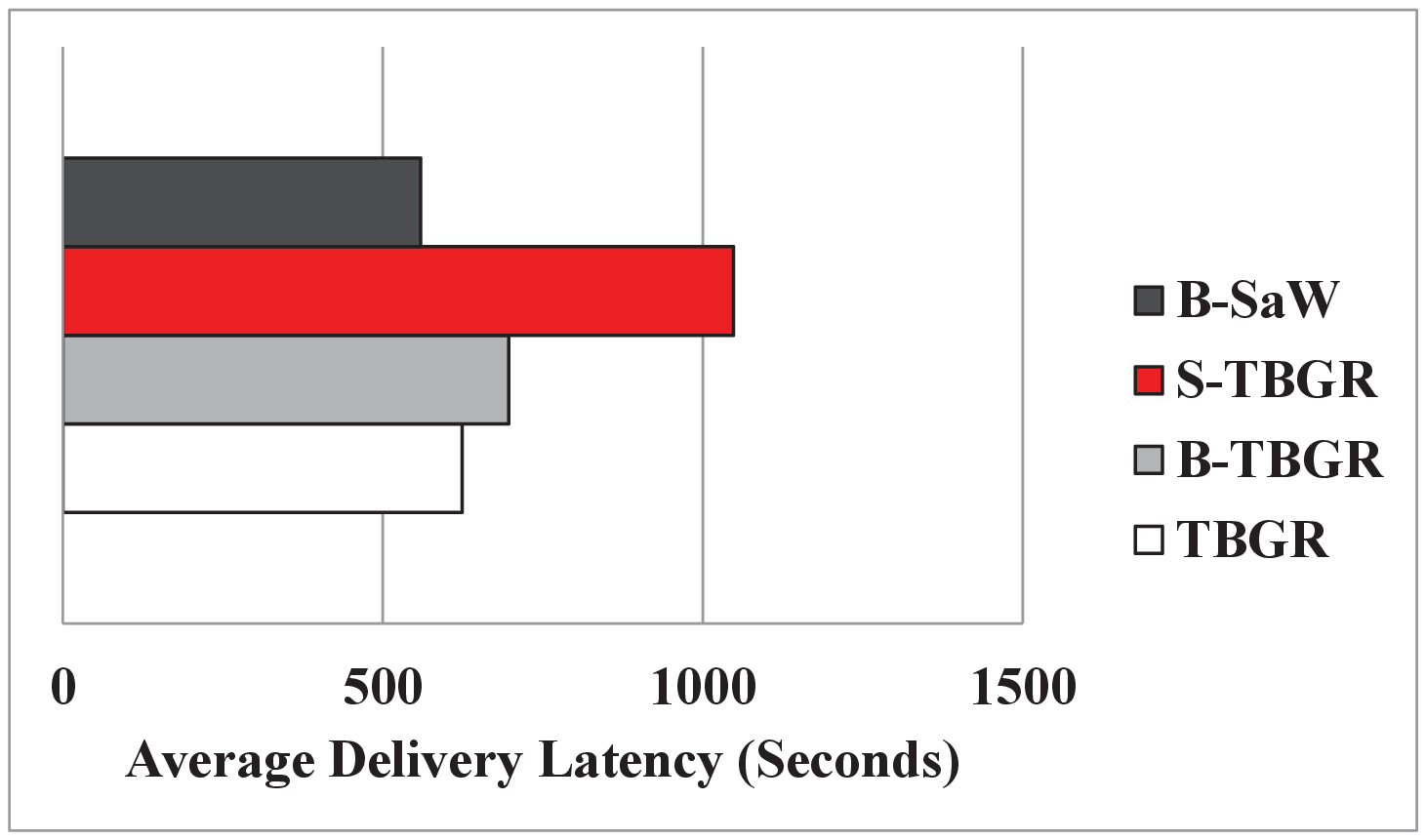}
    }
    \subfigure[Overhead Ratio]
    {
  \centering
    \label{f33}
    \includegraphics[width=5.7cm,height=3.3cm]{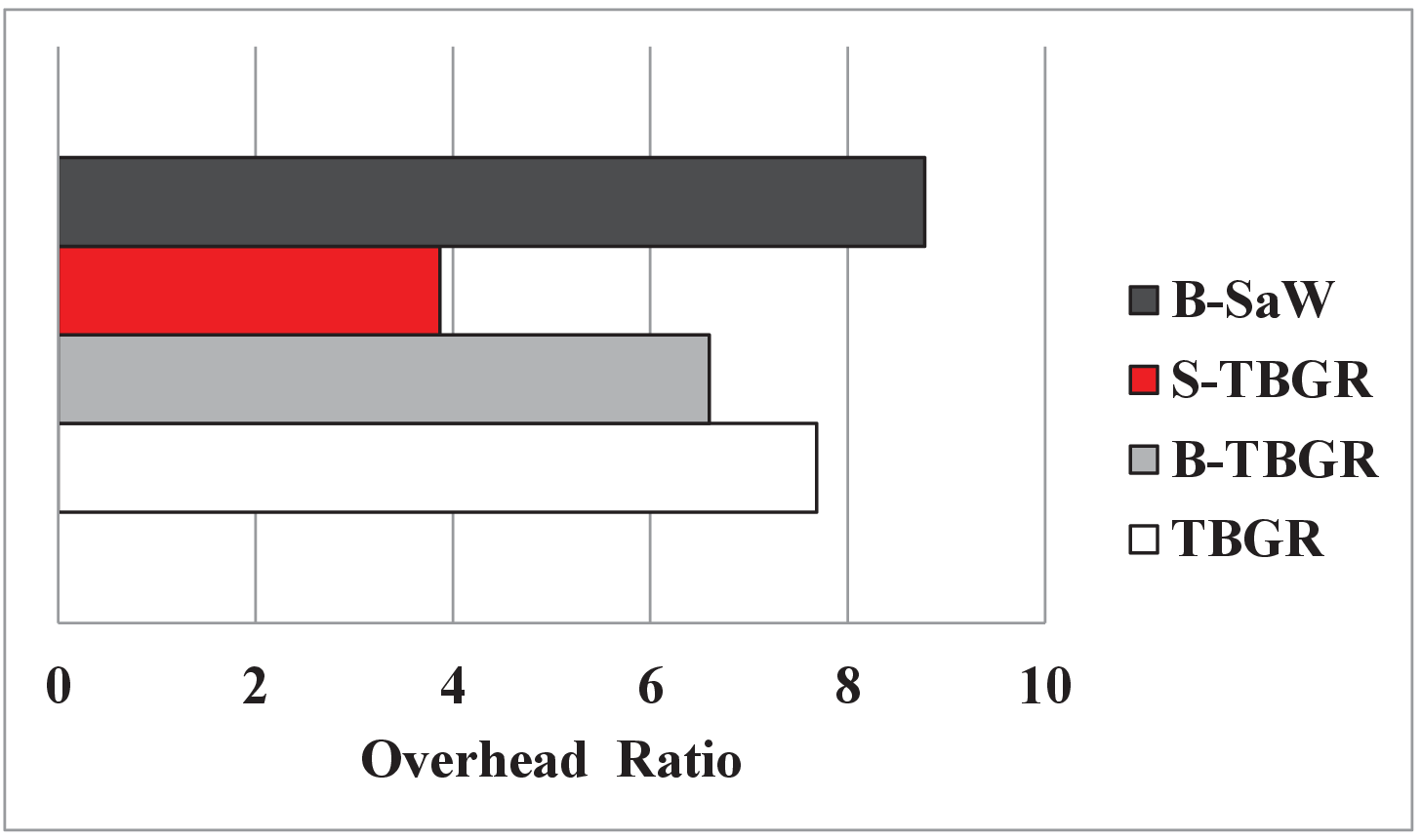}
    }
  \caption{Influence of Handling the Local Maximum Problem}\vspace{-10pt}
\end{figure*}

The discussed properties of TBGR are claimed through the following results, where evaluation metrics are explained as follows:
\begin{itemize}
\item \textbf{Delivery Ratio}: It is the ratio between the number of messages delivered and the total number of messages generated.

\item \textbf{Average Delivery Latency}: It is the average delay for a message to be delivered from the source node to its destination.

\item \textbf{Overhead Ratio}: It is the ratio between the number of relayed messages (excluding the delivered messages) and the number of delivered messages.
\end{itemize}

\subsubsection{Influence of Message Lifetime}
Here, S-SaW \cite{4430784} is the upper bound for S-TBGR and S-ABGR following Property \ref{t1}.
This is because S-SaW does not address candidate node selection, and just replicates $(L-1)$ message copies with $(C_M=1)$ copy ticket distributed each time.
In Fig.\ref{f11}, we observe that S-TBGR outperforms S-ABGR in terms of delivery ratio.
This follows Property \ref{t1} where the limitation of routing decision (in terms of inconsistent moving direction) is overcome.
Note that, DD \cite{1026005} performs as the lower bound since it only relays a message when the destination is in proximity.
Following Property \ref{t2}, S-TBGR achieves a lower delivery delay than S-ABGR in Fig.\ref{f12}, due to that the former replicates $(L-1)$ message copies faster.
In Fig.\ref{f13}, S-TBGR achieves a decreased overhead ratio, compared to S-ABGR suffering from an increased overhead ratio.
This implies the advantage of S-TBGR lets the best (rather than a better) relay node to keep on distributing message copies.
Note that the overhead ratio of DD is 0, because it does not involve intermediate nodes in relay process.

\subsubsection{Influence of Other Conditions}
We also show the delivery ratio based on varied moving speed and transmission range (with 20 minutes TTL configured) in Fig.\ref{f21} and Fig.\ref{f22} respectively, all schemes show similar trends compared with their previous results.
In Fig.\ref{f23} the routing performance is affected by the nodal activity (temporary waiting time).
Here,  S-TBGR outperforms S-ABGR following Property \ref{t5}.

\subsubsection{Influence of Handling the Local Maximum Problem}
The comparison between S-TBGR and B-TBGR is shown in Fig.\ref{f31}, Fig.\ref{f32} and Fig.\ref{f33} following Property \ref{t3}.
Although B-SaW is the upper bound for TBGR, TBGR achieves less than 19.2$\%$ overhead ratio, whereas only with less than 1$\%$ delivery ratio of B-SaW.
In spite of achieving a higher 11.6$\%$ delivery delay, we conclude the efficiency of TBGR by referring to a delay tolerant application for data collection.
Here, we claim Property \ref{t4} that TBGR achieves a comparable delivery ratio, while with a shorter delivery delay and lower overhead.

\section{Design of TBHGR}
\begin{figure}[htbp]
\begin{center}
\includegraphics[scale=0.35]{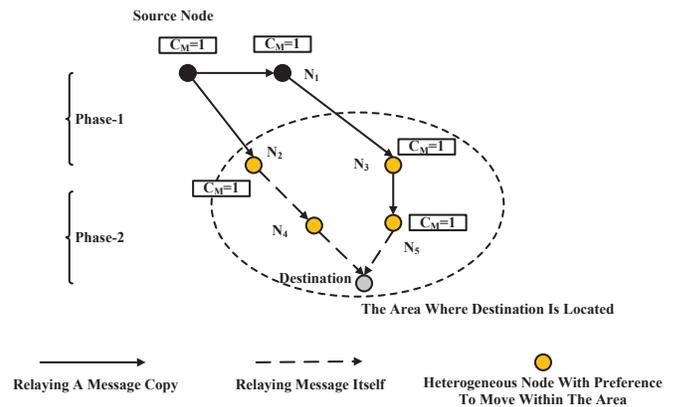}
\caption{An Overview of TBHGR}\vspace{-10pt}
\label{exp}
\end{center}
\end{figure}
In previous section, we introduced TBGR and discussed its properties based on the RWP mobility model where nodal mobility is homogeneous.
However, the realistic mobility model normally comprises heterogeneity.
For example, we can envisage a VSNs scenario where drivers often tend to travel within certain areas due to their visiting preference.
In such case, the mobility of vehicles would be limited to given areas for the majority of time.
Here, if only relying on the homogeneous nodal mobility, the potential for message delivery to a destination located at other area would be degraded. This is mainly because a message with $(C_M>1)$ copy tickets may never be replicated, to those nodes heading to the area where the destination is located.

Based on the design of TBGR and additional notations in TABLE \ref{table2}, as an example shown in Fig.\ref{exp}, the routing decision in The-Best-Heterogeneity-based-Geographic-Relay (TBHGR) is decoupled into two phases depending on the value of $C_M$.
The Phase-1 relays message copies with $(C_M>1)$ copy tickets towards the area where the destination is located, while the Phase-2 relays message copies with $(C_M=1)$ copy ticket within this area.

\subsection{Phase-1 of TBHGR}
\begin{table}[htbp]\scriptsize
\caption{List of Additional Notations Defined in TBHGR}\vspace{-10pt}
\label{configure}
\centering
\begin{tabular}{|p{0.7cm}|p{6.6cm}|}
\hline
$D^{\prime}_{j,d}$ & Projected distance from $N_j$ to $N_d$\\\hline
$W$ & Time window for estimating the projected distance\\\hline
$V^D_M$ & Heuristic distance based threshold value used for recoding $D^{\prime}_{j,d}$\\\hline
$\Psi_j$ & A set in $N_j$ recording the list of encountered nodes in the past\\\hline
$U_M$ & Utility of $M$ \\\hline
\end{tabular}
\label{table2}
\end{table}

As illustrated in Algorithm \ref{spray}, a message with $(C_M>1)$ copy tickets is geographically relayed as follows:
\subsubsection{$The \left(\phi_{j,d}<\frac{\pi}{2}\right)$ Case}
\begin{algorithm}[htbp]\footnotesize
\caption{Phase-1 of TBHGR}
\label{spray}
\begin{algorithmic}[1]
\STATE set $(C_M=L)$
\STATE set $V^T_M$ with an infinitely large value
\STATE set $V^D_M$ with an infinitely large value
\FOR{each encounter between $N_i$ and $N_j$}
\FOR{each $M$ carried by $N_i$}
\IF {$N_j$ already has a copy of $M$}
\STATE update $V^T_M$ and $V^D_M$, for $M$ including its copy carried by both $N_i$ and $N_j$
\ELSIF {$(C_M>1)\cap(S_j\neq0)$}
\IF {$\left(V^T_M>\frac{D_{j,d}-R}{S_j\times\cos\phi_{j,d}}\right)\cap\left(\phi_{j,d}<\frac{\pi}{2}\right)\cap\left(N_d\in\Psi_j\right)$}
\STATE update $V^T_M$ towards $\frac{D_{j,d}-R}{S_j\times\cos\phi_{j,d}}$
\STATE replicate $M$ to $N_j$ with $\frac{C_M}{2}$ copy tickets
\STATE keep $\left(C_M-\frac{C_M}{2}\right)$ copy tickets for $M$ in $N_i$
\ELSIF {$\left(V^T_M>T^{ini}_M-T^{ela}_M\right)\cap\left(\phi_{j,d}<\frac{\pi}{2}\right)$}
\STATE replicate $M$ to $N_j$ with $\frac{C_M}{2}$ copy tickets
\STATE keep $\left(C_M-\frac{C_M}{2}\right)$ copy tickets for $M$ in $N_i$
\ELSIF {$\left(V^D_M>D^{\prime}_{j,d}\right)\cap\left(\phi_{j,d}\geq\frac{\pi}{2}\right)$}
\STATE update $V^D_M$ towards $D^{\prime}_{j,d}$
\STATE replicate $M$ to $N_j$ with $(C_M=1)$ copy ticket
\STATE keep $(C_M-1)$ copy tickets for $M$ in $N_i$
\ENDIF
\ENDIF
\ENDFOR
\ENDFOR
\end{algorithmic}
\label{al1}
\end{algorithm}

Here, each mobile node maintains a set $\Psi$ to record a list of its encountered destinations in the past, where $(N_d\in\Psi_j)$ means $N_j$ had an encounter with $N_d$.
By jointly considering $(N_d\in\Psi_j)$ and conditions \eqref{con2}, we have conditions \eqref{sc} such that a message copy is replicated to $N_j$, with $\frac{C_M}{2}$ copy tickets distributed.
This follows the message relaying process from source node to $N_2$ in Fig.\ref{exp}.
\begin{equation}\footnotesize
\left(V^T_M>\frac{D_{j,d}-R}{S_j\times\cos\phi_{j,d}}\right)\cap\left(\phi_{j,d}<\frac{\pi}{2}\right)\cap\left(N_d\in \Psi_j\right)
\label{sc}
\end{equation}
Following the presentation between lines 9 and 12 in Algorithm \ref{spray}, this executes that a message copy with $\frac{C_M}{2}$ copy tickets is replicated to the node which has preference to visit the destination.
The motivation is that the node that did not meet $N_d$ in the past, may replicate message copies to other nodes without visiting preference.
In the worst case, $N_j$ will never meet $N_d$, even if it knows the location of $N_d$.
As such, some of copies would not contribute to delivery, but result in additional routing overhead.

Furthermore, given the local maximum problem that $N_j$ does not meet conditions \eqref{sc}, a message copy is conditionally replicated with $\frac{C_M}{2}$ copy tickets distributed.
This happens only if the message will expire before $\frac{D_{i,d}-R}{S_i\times\cos\phi_{i,d}}$ as recorded in $V^T_M$.
Here, considering the case $N_j$ has encountered $N_d$ in the past but with $\left(V^T_M\leq\frac{D_{j,d}-R}{S_j\times\cos\phi_{j,d}}\right)$, the condition $(N_d\in\Psi_j)$ is omitted as presented between lines 13 and 15 in Algorithm \ref{spray}.
Such operations increase the possibility to find a node that meets conditions \eqref{sc}, by searching from heterogeneous nodes with different visiting preferences.

\subsubsection{The $\left(\phi_{j,d}\geq\frac{\pi}{2}\right)$ Case}
As presented between lines 16 and 19 in Algorithm \ref{spray}, a distance based geographic utility metric is applied if the encountered node $N_j$ is moving away from $N_d$.
Shown in Fig.\ref{dgrleave}, the projected distance $D^{\prime}_{j,d}$ measured from $N_d$ to the expected location of $N_j$ is calculated as:
\begin{equation}\footnotesize
D^{\prime}_{j,d}=D_{j,d}-W\times\cos\phi_{j,d}\times S_j-R
\label{projecteddistance}
\end{equation}
Given the nodal movement within a time window $W$, the condition $\left(D^{\prime}_{i,d}>D^{\prime}_{j,d}\right)$ implies that $N_j$ is closer to $N_d$, thus $N_j$ is selected to prevent message relaying away from $N_d$.

\begin{figure}[htbp]
\begin{center}
\includegraphics[scale=0.43]{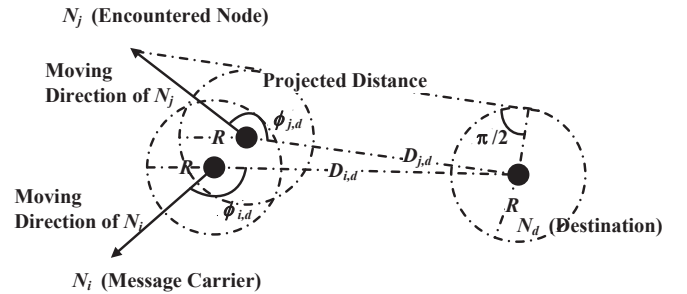}
\caption{Illustration of Projected Distance}\vspace{-10pt}
\label{dgrleave}
\end{center}
\end{figure}

Based on the previous discussion on DF, we define an additional threshold value $V^D_M$\footnote[6]{As an additional flag defined in each message, $V^D_M$ is also initialized with an infinitely large value and updated towards the smallest value, similar to $V^T_M$.} to record the value of $D^{\prime}_{j,d}$, then convert the condition $\left(D^{\prime}_{i,d}>D^{\prime}_{j,d}\right)$ into $\left(V^D_M>D^{\prime}_{j,d}\right)$.
Note that the investigation of DF herein also includes the case given $\left(\phi_{j,d}\geq\frac{\pi}{2}\right)\cap\left(\phi_{i,d}<\frac{\pi}{2}\right)$.
Thanks to updating $V^{D}_M$, the relay node will be gradually selected as the one, either with the $\phi_{j,d}$ close to $\frac{\pi}{2}$ or with the smallest $S_j$.
As such, the value of $(W\times\cos\phi_{j,d}\times S_j)$ is quite small, considering the sparse network density and highly dynamic movement.
Even assuming $(\phi_{j,d}=\frac{\pi}{2})$, the relay node which is closer to $N_d$ will be selected, given that $V^{D}_M$ has been updated to $(D_{j,d}-0-R)$.

Since the message has been relayed with $\frac{C_M}{2}$ copy tickets in case of $\left(\phi_{j,d}<\frac{\pi}{2}\right)$, here a message copy is replicated with only $(C_M=1)$ copy ticket given the condition $\left(\phi_{j,d}\geq\frac{\pi}{2}\right)$.
This is because it is beneficial to distribute more copy tickets for replicating towards destination to reduce delivery delay, rather than doing so for replicating away from destination.
Here, the condition $(N_d\in \Psi_j)$ is not limited, for which the motivation is to faster distribute copy tickets to heterogeneous nodes that would encounter $N_d$ in the future.

Different from the intention in case of $\left(\phi_{j,d}<\frac{\pi}{2}\right)$, it is unnecessary to further handle the local maximum problem given $\left(V^D_M\leq D^\prime_{j,d}\right)$, as the relaying direction is already away from the destination.
Since this message with the $\left(C_M=1\right)$ copy ticket cannot be replicated, those nodes carrying this message but have not encountered $N_d$, will perform the following Algorithm \ref{relay} in Phase-2.

\subsection{Phase-2 of TBHGR}
\begin{algorithm}[htbp]\footnotesize
\caption{Phase-2 of TBHGR}
\label{relay}
\begin{algorithmic}[1]
\FOR{each encounter between $N_i$ and $N_j$}
\FOR{each $M$ carried by $N_i$}
\IF {$N_j$ already has a copy of $M$}
\STATE update $V^T_M$ and $V^D_M$ for $M$ including its copy carried by both $N_i$ and $N_j$
\ELSIF {$(C_M=1)\cap(S_j\neq0)$}
\IF {$\left(V^T_M>\frac{D_{j,d}-R}{S_j\times\cos\phi_{j,d}}\right)\cap\left(\phi_{j,d}<\frac{\pi}{2}\right)\cap(N_d\in\Psi_j)$}
\STATE update $V^T_M$ towards $\frac{D_{j,d}-R}{S_j\times\cos\phi_{j,d}}$
\STATE replicate $M$ to $N_j$
\STATE delete $M$ in $N_i$
\ELSIF {$\left(V^T_M>T^{ini}_M-T^{ela}_M\right)\cap(N_d\in\Psi_j)$}
\STATE replicate $M$ to $N_j$
\ENDIF
\ENDIF
\ENDFOR
\ENDFOR
\end{algorithmic}
\label{al1}
\end{algorithm}
As presented between lines 6 and 9 in Algorithm \ref{relay}, the message with $(C_M=1)$ copy ticket is relayed without keeping the originally carried message, towards the encountered node meeting conditions \eqref{sc}.
Furthermore, the condition $\left(\phi_{j,d}\geq\frac{\pi}{2}\right)$ is not taken into account for routing decision, since relaying such a message away from its destination shall result in redundant routing overhead.
This follows the message relaying process from $N_2$ to $N_4$ in Fig.\ref{exp}.

Here, the value of $V^D_M$ is still updated, as other copies with $(C_M>1)$ copy tickets may still be relayed following conditions \eqref{projecteddistance}.
Since the message relaying is unidirectional, the message with $(C_M=1)$ copy ticket is always relayed to the node
meets condition \eqref{sc}.
This is also similar to the design of achieving loop free in traditional network, by setting a maximum value, such as hop count to prevent further message relay.

Besides, we consider to handle the local maximum problem as presented between lines 10 and 11 in Algorithm \ref{relay}.
This follows the message relaying process from $N_3$ to $N_5$ in Fig.\ref{exp}.
By relying on the contribution in Section \uppercase\expandafter{\romannumeral3}, conditions $\left(V^T_M>T^{ini}_M-T^{ela}_M\right)\cap\left(N_d\in \Psi_j\right)$ will let $N_i$ generate a message copy to $N_j$.
The motivation of considering $\left(N_d\in \Psi_j\right)$, is to relay an additional copy to the encountered node which has met the destination.
This happens only if the shortest time duration to intersect $N_d$, as the updated value of $V^T_M$, is longer than the expiration deadline $\left(T^{ini}_M-T^{ela}_M\right)$.
Note that the condition $\left(\phi_{j,d}<\frac{\pi}{2}\right)$ is not limited herein.
This increases the possibility for $N_j$ to encounter other nodes satisfying conditions \eqref{sc}.

\subsection{Computation Complexity of Routing Decision}
The computation complexity follows $O(\log_2L)\leq O(TBHGR)\leq O(L)$, in order to have $L$ message copies with $(C_M=1)$ copy ticket.
This is because that, on the one hand, TBHGR requires $(L-1)$ times to have $L$ message copies with $(C_M=1)$ copy ticket, if the routing decision is processed by the $\left(\phi_{j,d}\geq\frac{\pi}{2}\right)$ case.
Here, relay nodes are selected with the descending order of $V^D_M$.
On the other hand, $O(\log_2L)$ is required if processed via the $\left(\phi_{j,d}<\frac{\pi}{2}\right)$ case, where relay nodes are selected with the descending order of $V^T_M$.

\subsection{Message Transmission}
As $\left(T^{ini}_M-T^{ela}_M\right)$ is calculated as the remaining message lifetime, a positive value of $\left(\left(T^{ini}_M-T^{ela}_M\right)-V^T_M\right)$ implies the message can be delivered before its expiration deadline, based on the shortest time duration $\frac{D_{j,d}-R}{S_j\times\cos\phi_{j,d}}$ recorded in the past.
Note that since $V^T_M$ will not be updated in case of $\left(\phi_{j,d}\geq\frac{\pi}{2}\right)$, $\frac{\left(T^{ini}_M-T^{ela}_M\right)-V^T_M}{\left(T^{ini}_M-T^{ela}_M\right)}$ is calculated to consider the probability that each message copy is successfully delivered within the expiration deadline, as $P_{(L-1)}$ discussed in Property \ref{t1}.

Next, by considering the copy ticket $C_M$, the message with a larger value of $C_M$ implies that it is still allowed for replication.
Here, since a larger number of message copies will increase the message delivery probability, as discussed in Property \ref{t1}, the message with a larger value of $C_M$ has a higher potential to be delivered.

Based on the above discussion, the utility $U_M$ to qualify each message is defined as:
\begin{equation}\footnotesize
U_M=1-\left(1-\frac{\left(T^{ini}_M-T^{ela}_M\right)-V^T_M}{\left(T^{ini}_M-T^{ela}_M\right)}\right)^{C_M}
\end{equation}
Therefore, $U_M$ implies the potential for message delivery within its remaining lifetime, considering the nodal mobility to intersect destination as well as its number of copies $C_M$.
Here, the message with the largest $U_M$ is managed with the highest priority for transmission.
The key insight behind this is to transmit a message which is still able for replication, while with a long TTL.

\subsection{Buffer Management}
If the remaining buffer space is close to an overflow state, $N_j$ deletes its stored messages from the one with the lowest $U_M$.
This operation continues until the space for those received messages from $N_i$ is allocated.
In addition, when any message is finally delivered, an ACK\footnote[7]{Compared with data message, the ACK message is with quite small size that only contains nodal ID (like a ``string"). Therefore, the bandwidth and buffer space consumed by ACK can be ignored. In fact, the real implementation of ACK is operated as a table list in each node, to record the ID of the delivered messages, rather than recording a whole message copy in this list.} information containing the ID of this message will be generated by destination.
This information will be exchanged when any node encounters destination, and recorded locally in a table list.
Then, an intermediate node has the knowledge of this message ID will check its buffer, and delete the copy of this message.
A fast dissemination of this knowledge benefits from a high mobility.

\subsection{Performance Evaluation}
The scenario is based on the downtown area of Helsinki city with 4500$\times$3400 $m^2$ area.
This is important to evaluate geographic routing, since nodal speed, direction as well as distance are considered in TBHGR.
The moving speed of mobile nodes is randomly chosen from [10$\sim$50] km/h, and their waiting time varies between [0$\sim$240]s.
The communication technique is set with 30m transmission range and 4 Mbit/s bandwidth, considering as the low power WiFi technology.

Envisioning for a heterogeneous network, we also assign four types of Points Of Interests (POIs), by default 3 destinations are deployed shown in Fig.\ref{scenario}, 30 mobile nodes of each group are allocated to each type of POI, with 0.8 probability moving around the POIs and 0.2 probability roaming in the entire network.
Note that mobile nodes will encounter more likely and frequently due to a high interests to a type of POI.
Those areas consist of corresponding POIs, while we plot them as ellipses for simplicity.
Note that even though there is intersection between Area-2 and Area-4, delivering messages from Area-2 to Area-4 is still affected by the nodal movement interest.
\begin{figure}[htbp]
\begin{center}
\includegraphics[scale=0.45]{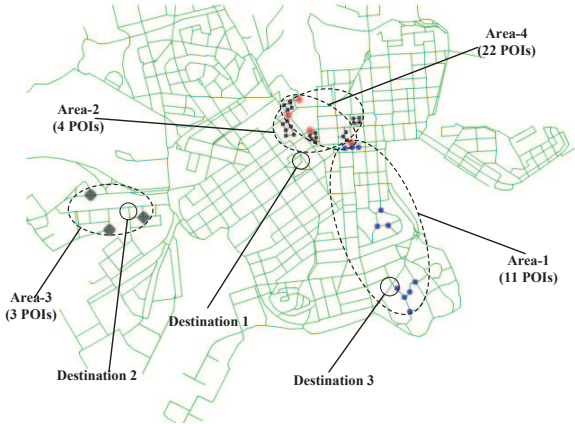}
\caption{The Helsinki City Scenario}\vspace{-10pt}
\label{scenario}
\end{center}
\end{figure}

In addition to B-SaW \cite{4430784} evaluated in Section \uppercase\expandafter{\romannumeral3}, GeoSpray \cite{geospray} is a geographic routing scheme envisioning for VSNs scenario while assuming homogeneous nodal mobility.
Note that GeoSpray limits the number of copies a message can be replicated.
Besides, EBRR \cite{7055214} and EBSR \cite{7055214} are topological routing schemes proposed for VSNs comprising heterogeneous nodal mobility, where EBSR also limits the number of message copies up to $L$.

The simulation runs include an initial warm-up period of 3600s before message generation.
Messages are randomly generated from all mobile nodes for every 30s, with 60 minutes lifetime and 1MB size such that a transmission contention would exist.
The nodal buffer space is set to be 40 MB.
The $W$ in TBHGR is fixed with 5s, while the $L$ for TBHGR, B-SaW, EBSR and GeoSpray is configured as 12.
This also follows \cite{4430784} that choosing $L$ equals to around 10$\%$ number of mobile nodes in a network.
To measure the full activity of a network, the message generation ends before 39600s with an additional 3600s allowed to consume the unexpired messages.
Since a number of parameters related to routing decision have been discussed in Section \uppercase\expandafter{\romannumeral3} under homogeneous scenario, we here turn to the factor of mobility heterogeneity concerning the visiting preference and distribution of destinations.
\begin{figure*}[htbp]
  \centering
  \subfigure[Delivery Ratio]
  {
  \centering
    \label{mi1}
    \includegraphics[width=5.7cm,height=3.3cm]{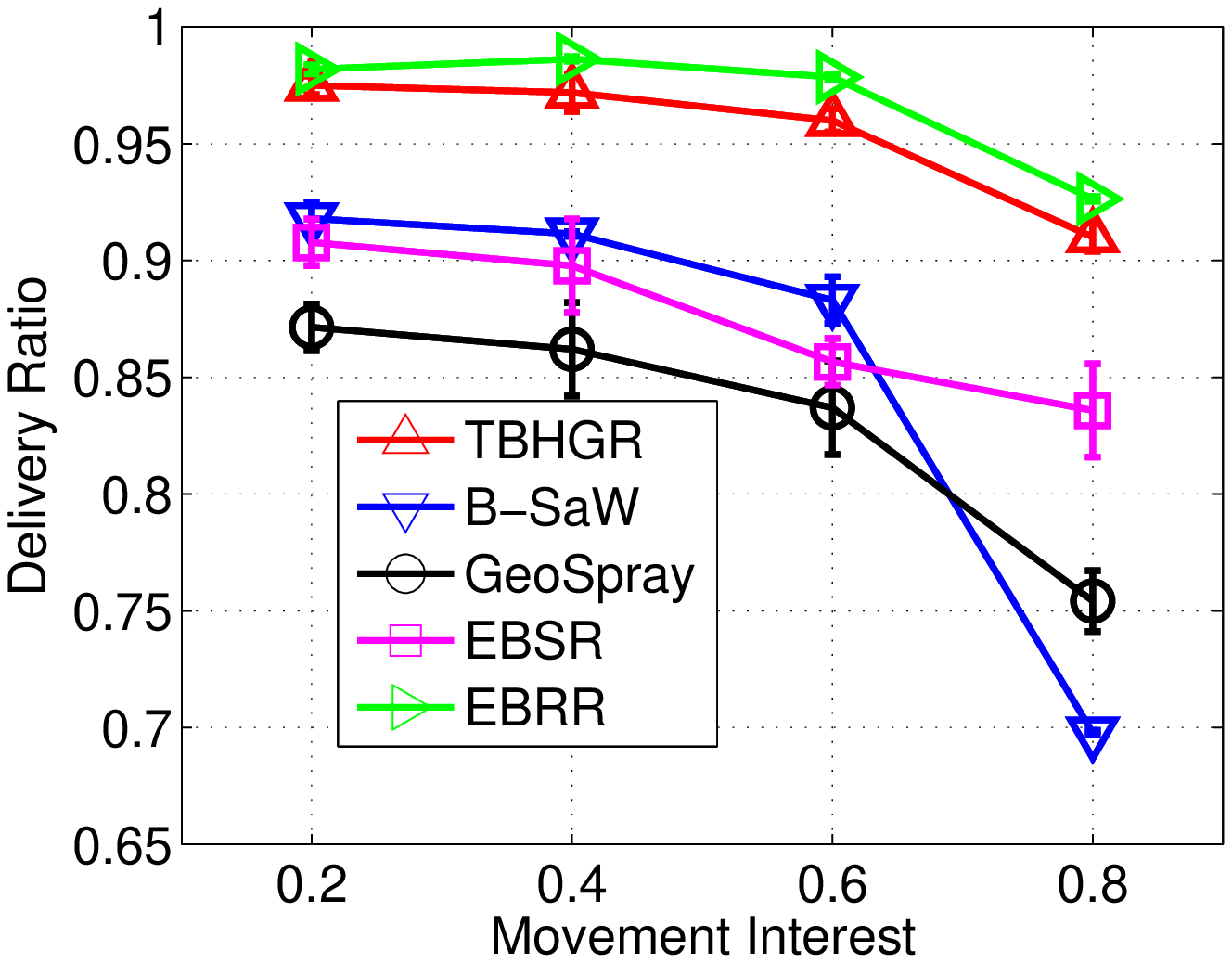}
    }
  \subfigure[Average Delivery Latency]
  {
  \centering
    \label{mi2}
    \includegraphics[width=5.7cm,height=3.3cm]{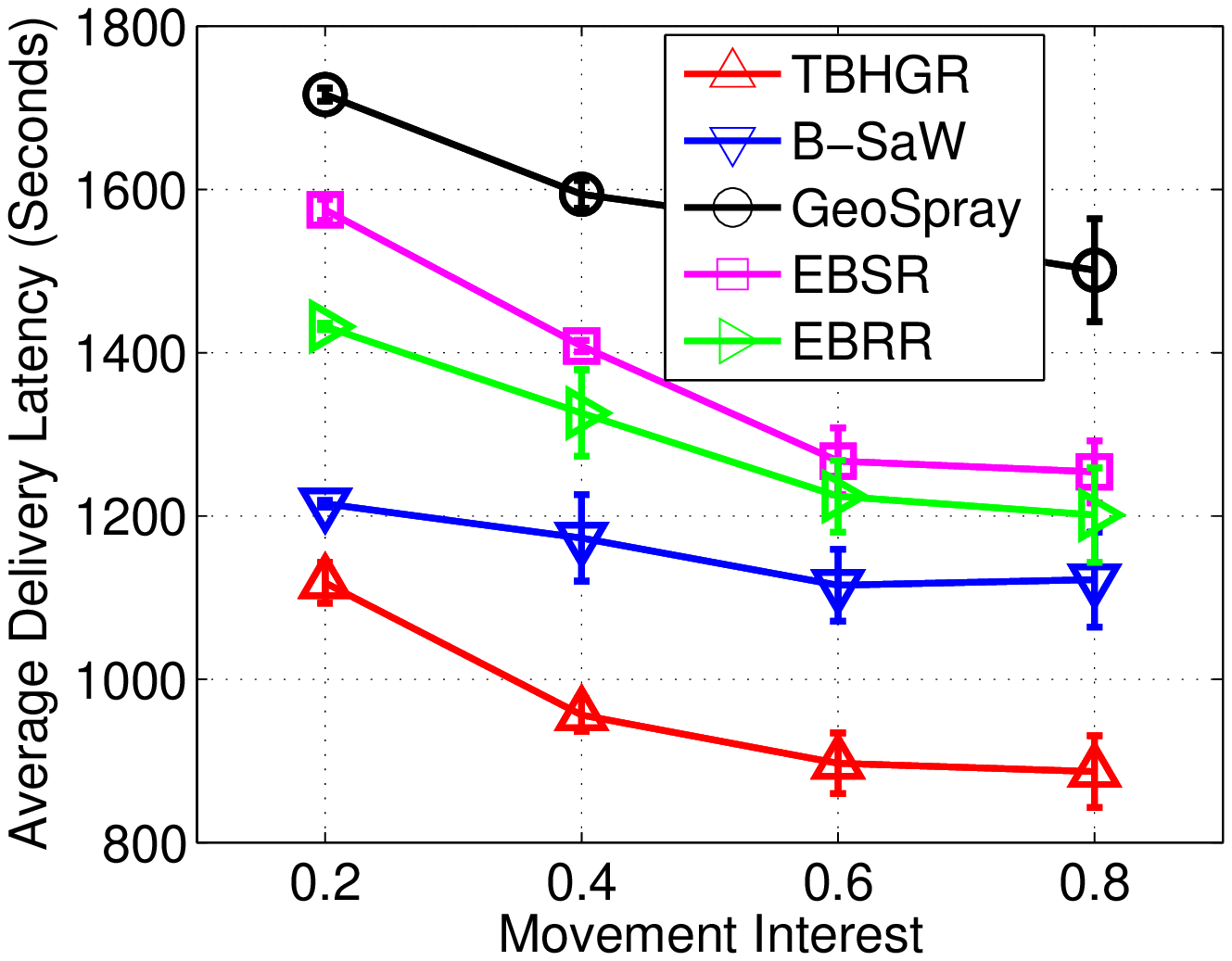}
    }
    \subfigure[Overhead Ratio]
    {
  \centering
    \label{mi3}
    \includegraphics[width=5.7cm,height=3.3cm]{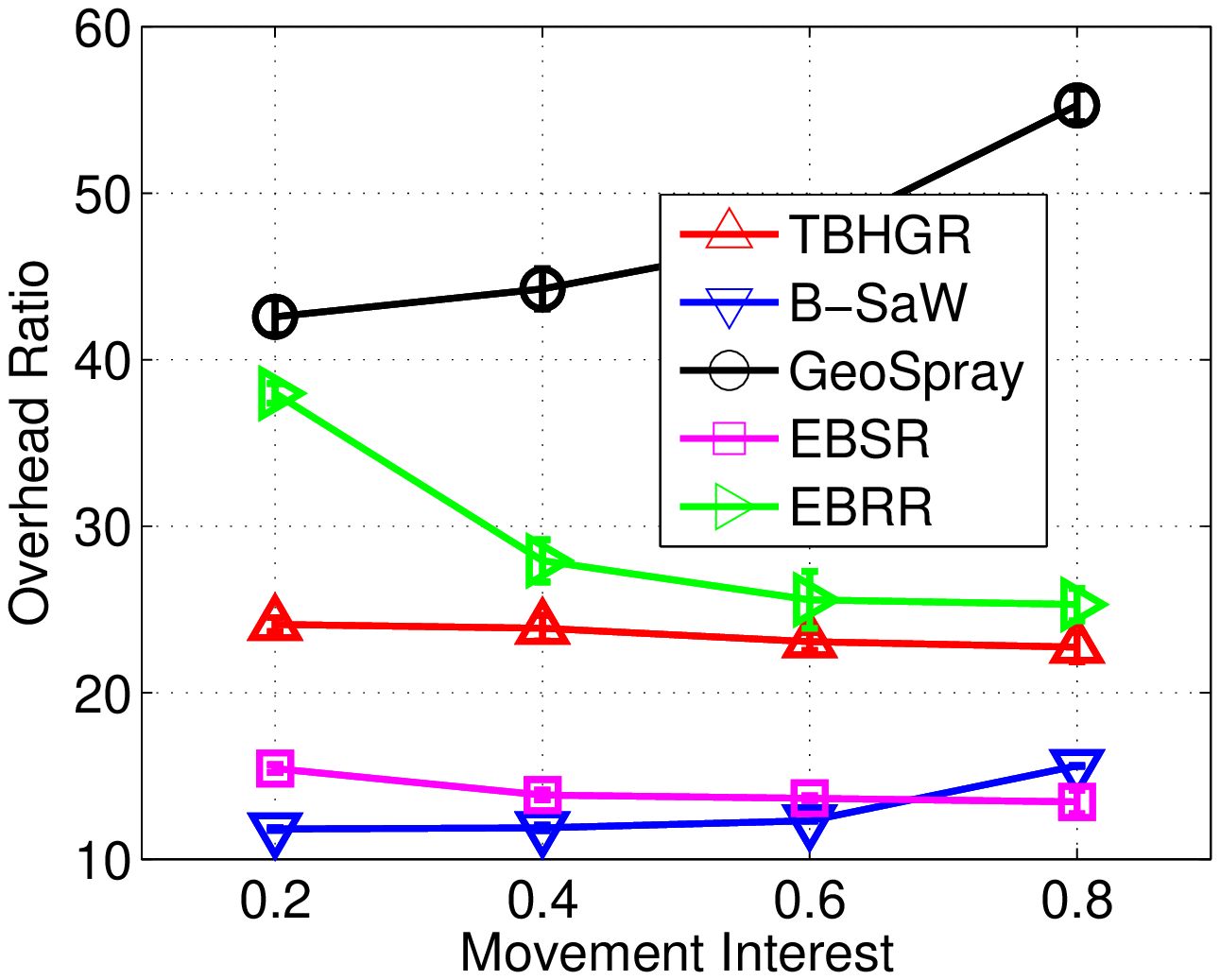}
    }
  \caption{Influence of Movement Interest}\vspace{-10pt}
\end{figure*}

\subsubsection{Influence of Movement Interest}
In Fig.\ref{mi1}, all schemes perform worse in case of 0.8 movement interest.
This is because nodes are highly possible to move around dedicated POIs of areas, rather than just roaming across an entire network.
We also observe that B-SaW suffers more from performance degradation, with increased movement interest.
This is because that it only performs well in case of homogeneous nodal mobility (can be partially reflected by the case with 0.2 movement interest).
However, if nodal mobility is heterogeneous, B-SaW may distribute message copies to $(L-1)$ nodes which never meet destination.
Here, TBHGR achieves the highest delivery ratio comparing with other schemes which initially limit the number of message copies (e.g., B-SaW, EBSR, GeoSpray).

In Fig.\ref{mi2}, all schemes are with a decreased average delivery latency, primarily due to delivering a less number of messages.
Here, TBHGR achieves the lowest value, owing to geographically relaying messages under a highly mobile scenario.
Particularly, TBHGR is more efficient than EBRR (which utilizes topological utility metric), while TBHGR already achieves a close delivery ratio of EBRR.
Here, the main reason for delay improvement is the advantage of capturing geographic information over topological information, note that the latter is extensively fluctuated under such highly dynamic scenario.
Different from topological routing schemes like EBRR and EBSR (inevitably use obsolete historical encounter information), geographic routing schemes like TBHGR and GeoSpray rely on real-time geographic information for mobility prediction.
In DTNs, the large variation of network topology is the main reason degrades the performance of topological routing schemes.
Therefore TBHGR outperforms EBRR and EBSR.
It is worth noting that TBHGR also outperforms B-SaW, because that TBHGR addresses the heterogeneous nodal mobility.

The observation in Fig.\ref{mi3} shows that TBHGR achieves a decreased overhead ratio.
In contrast, GeoSpray is with an increased overhead ratio, mainly due to not considering heterogeneous mobility.
Due to the same reason, B-SaW brings an increased redundancy (but that does not effectively contribute to message delivery).
An important observation is that, not limiting the number of copies of a message for replication in case of low nodal movement interest, would result in much redundancy as performed by EBRR.
While, TBHGR and EBSR are with a smoothly increased overhead ratio, by however limiting the number of message copies.

\subsubsection{Influence of Distribution of Destinations}
Since previous results are shown given pre-deployed destinations, we further implement a location distribution function depending on the nodal movement interest.
Here, a certain number of coordinates are selected from the 40 POIs as already illustrated in Fig.\ref{scenario}.
Meanwhile, the same number of destinations are distributed with a distance variation to those points.
For example, the case with ``7 Des (0 Var)'' indicates the locations of 7 destinations are randomly selected from 40 POIs, without any distance variation.
The destinations with ``500 Var'' indicate their locations are with a minimum 500m away from the randomly selected POIs.
The underlying map scenario is formed by a number of coordinates, where a link between pairwise coordinates forms a route path.
The initial location distribution of a mobile relay can be one of the inherent map coordinates, or any place between pairwise coordinates (on a path).
Here, considering a destination with 500m distance variation might not be located on a map coordinate, we thus select a map coordinate for the destination which is with the closest distance to this given destination.
Note that this finalized location also meets the condition that the distance between destination's location and the given POI is larger than 500m.

In Fig.\ref{rd1}, results are shown in the case of 2 random seeds, namely Seed-1 and Seed-2.
On the one hand, we observe that the number of destinations has less influence on the performance of those routing schemes considering mobility heterogeneity (e.g., TBHGR, EBRR and EBSR).
On the other hand, the distance variation between a destination and POI mainly affects routing performance.
If a destination is far away from the POI, those routing schemes (e.g., B-SaW, EBSR) of which the message replication is limited by $L$, suffer more from performance fluctuation compared with EBRR.
In spite of this, TBHGR achieves a close delivery ratio of EBRR, but with a lower overhead.
Upon this, we claim the efficiency of TBHGR, and its tolerance for the distribution of destinations.

\begin{figure}[htbp]
\begin{center}
\includegraphics[width=9cm,height=5cm]{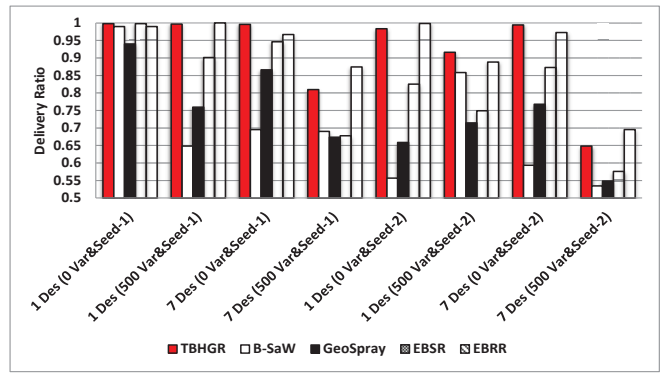}
\caption{Delivery Ratio vs Destinations Distribution}\vspace{-10pt}
\label{rd1}
\end{center}
\end{figure}

\section{Conclusion}
This article addressed the challenges of applying geographic routing in DTNs comprising nodal heterogeneous mobility.
We first presented TBGR and generalized its properties in terms of routing reliability under homogeneous scenarios, where nodal mobility is identical.
Upon this contribution, we next proposed TBHGR for heterogeneous scenarios where nodes have different visiting preference, with the consideration of message scheduling for transmission and storage.
Extensive results under the Helsinki city scenario envisioning for VSNs, with four types of POIs to form nodal heterogeneity, show the efficiency of TBHGR in terms of a low routing overhead and a lower delivery delay, which benefit from applying geographic information under highly dynamic scenarios.

\bibliographystyle{IEEEtran}
\bibliography{refer}
\begin{IEEEbiography}[{\includegraphics[width=1in,height=1.25in,clip,keepaspectratio]{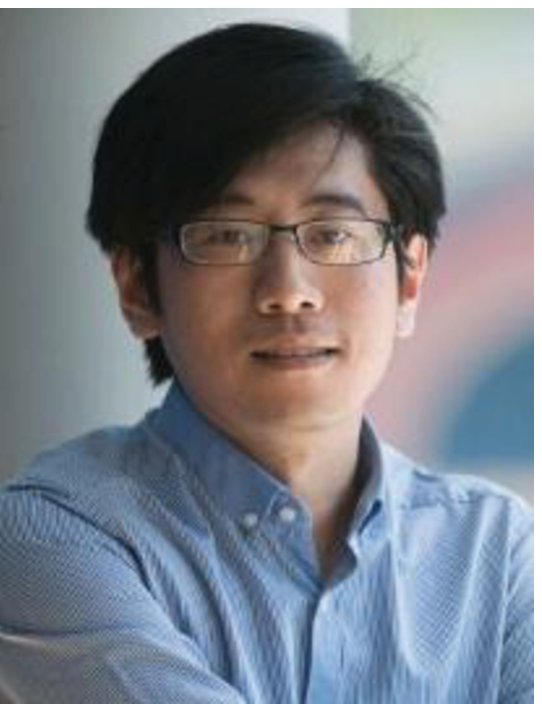}}]{Yue Cao} received his PhD degree the Institute for Communication Systems formerly known as Centre for Communication Systems Research, at University of Surrey, Guildford, UK in 2013. He is currently a Research Fellow at the ICS. His research interests focus on Delay/Disruption Tolerant Networks, Electric Vehicle (EV) communication, Information Centric Networking (ICN), Device-to-Device (D2D) communication and traffic offloading for cellular systems.
\end{IEEEbiography}

\begin{IEEEbiography}[{\includegraphics[width=1in,height=1.25in,clip,keepaspectratio]{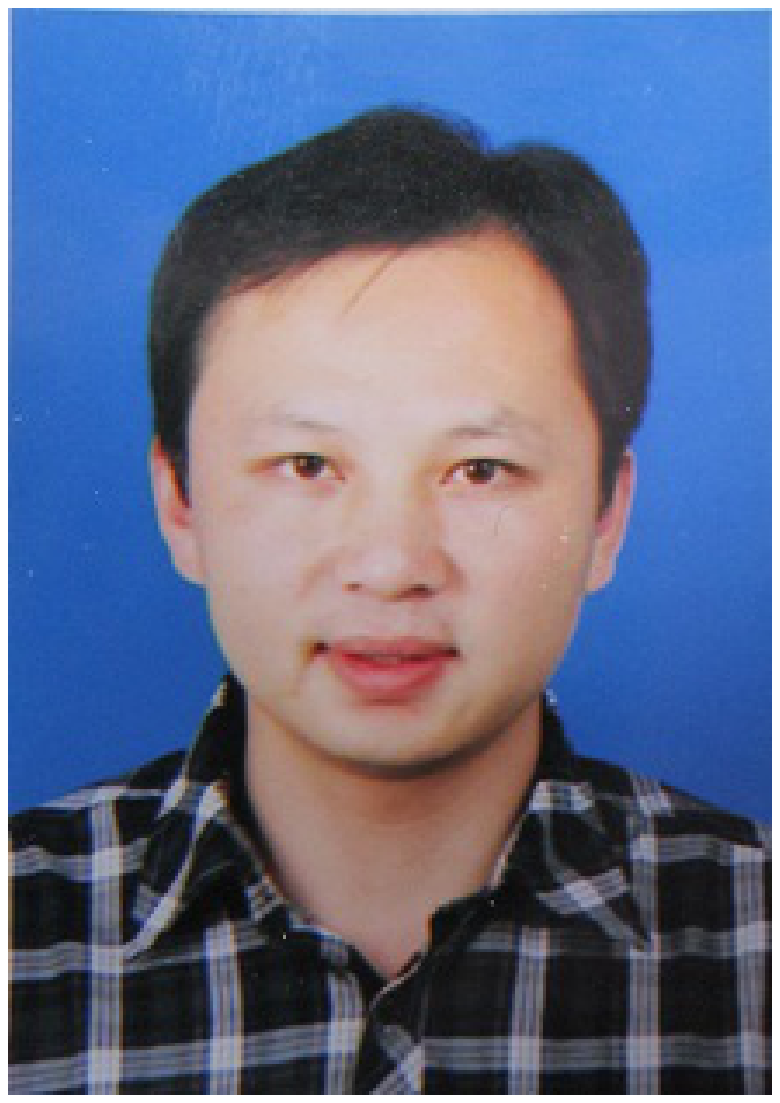}}]{Kaimin Wei} received the PhD degree in computer science and technology from Beihang University, Beijing, China, in 2014. He is currently an associate research professor at the School of Information Technology, Jinan University. His research interests focus on Delay/Disruption Tolerant Networks, Mobile Social Networks, and Cloud Computing.
\end{IEEEbiography}

\begin{IEEEbiography}[{\includegraphics[width=1in,height=1.25in,clip,keepaspectratio]{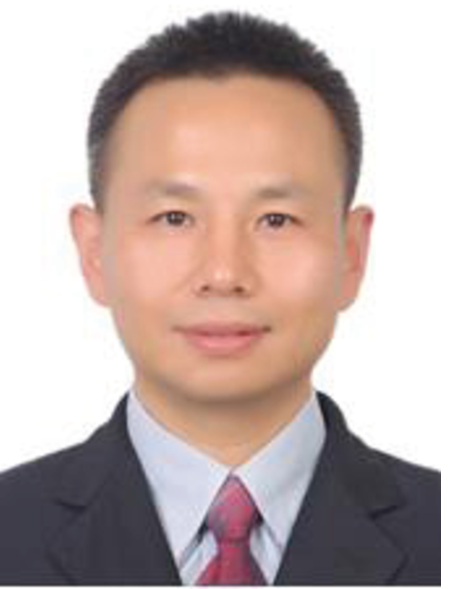}}]{Geyong Min} is a Professor of High Performance Computing and Networking in the Department of
Mathematics and Computer Science at the University of Exeter, UK. He received the PhD degree in Computing Science from the University of Glasgow, UK, in 2003. His research interests include Future Internet, Computer Networks, Wireless Communications, Multimedia Systems, Information Security, High Performance Computing, Ubiquitous Computing, Modelling and Performance Engineering.
\end{IEEEbiography}

\begin{IEEEbiography}[{\includegraphics[width=1in,height=1.25in,clip,keepaspectratio]{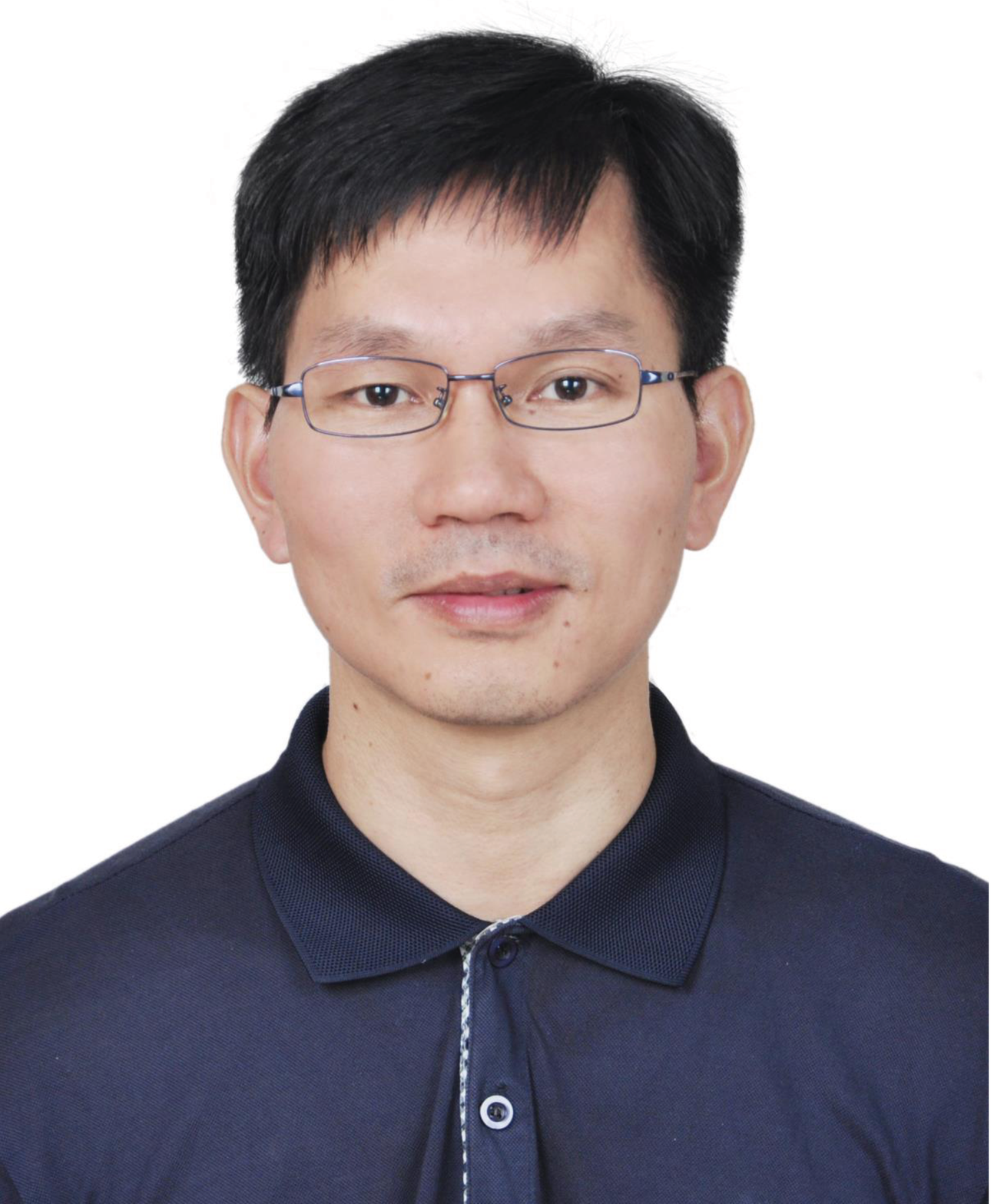}}]{Jian Weng} received the MSc and BSc degrees in computer science and engineering from South China University of Technology, in 2004 and 2000, respectively, and the Ph.D. degree in computer science and engineering from Shanghai Jiao Tong University, in 2008. From April 2008 to March 2010, he was a postdoc in the School of Information Systems, Singapore Management University. Currently, he is a professor and vice dean with the School of Information Technology, Jinan University. He has published more than 60 papers in cryptography conferences and journals, such as CRYPTO, EUROCRYPT, ASIACRYPT, TCC, PKC, CT-RSA, IEEE TDSC, IEEE TIFS, etc. He served as PC co-chairs or PC member for more than 20 international conferences.
\end{IEEEbiography}

\begin{IEEEbiography}[{\includegraphics[width=1in,height=1.25in,clip,keepaspectratio]{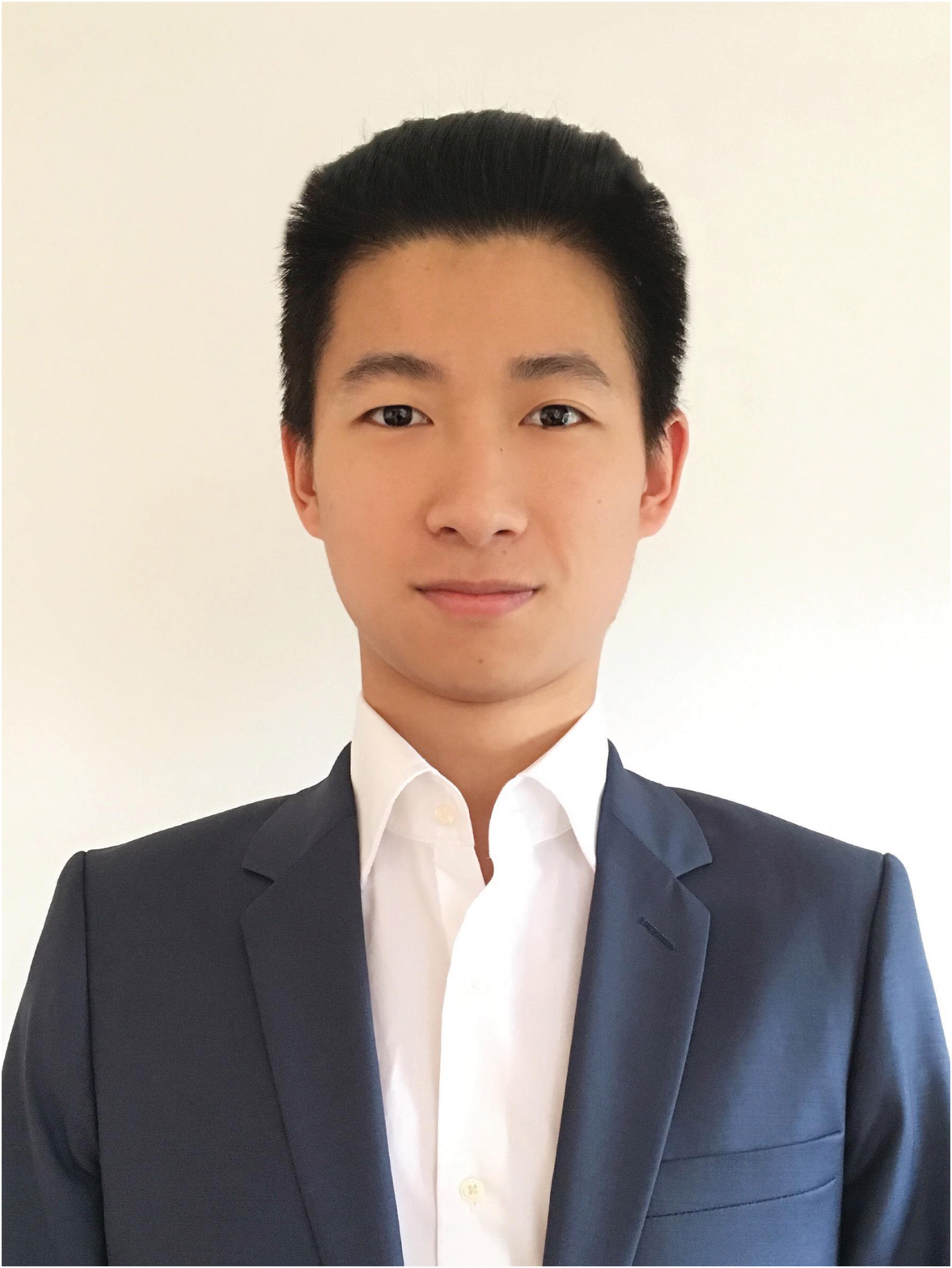}}]{Xin Yang} is currently a PhD Electronic Engineering candidate at Institute for Communication Systems, University of Surrey. He received his BEng Electronic Information Engineering degree from Harbin Institute of Technology, China, and MSc Communication Networks and Software from University of Surrey, UK, in 2013 and 2014 respectively. His current research area include Mobile Ad hoc Networks (MANETs), satellite networks and QoS routing.
\end{IEEEbiography}

\begin{IEEEbiography}[{\includegraphics[width=1in,height=1.25in,clip,keepaspectratio]{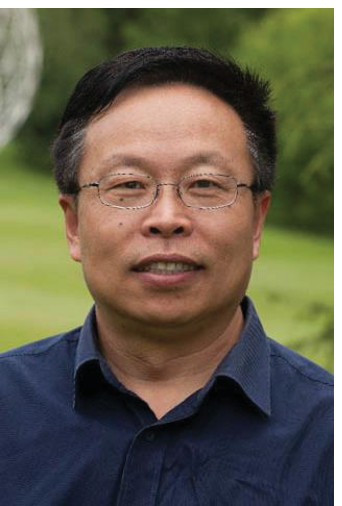}}]{Zhili Sun} (Chair of Communication Networking) is a Professor at the Institute for Communication Systems, University of Surrey, Guildford, UK. He obtained his PhD degree from Lancaster University, Lancaster, UK, in 1991. His research interests include wireless and sensor networks, satellite communications, mobile operating systems, traffic engineering, Internet protocols and architecture, QoS, multicast and security.
\end{IEEEbiography}
\end{document}